\documentclass[%
 reprint,
 amsmath,amssymb,
 aps,
]{revtex4-2}

\usepackage{graphicx}
\usepackage{dcolumn}
\usepackage{bm}
\usepackage{braket}
\usepackage{xcolor}
\usepackage{tikz}
\usepackage{ulem}
\usepackage[caption=false]{subfig}
\usepackage[version=4]{mhchem}
\usetikzlibrary{quantikz2}
\usepackage{multirow}
\usepackage{makecell}
\usepackage{amsthm}
\usepackage{soul}
\usepackage{cancel}

\usepackage{etoolbox} 

\DeclareMathOperator*{\argmin}{arg\,min}

\newtheorem{lemma}{Lemma}

\begin{document}

\preprint{APS/123-QED}

\title{Efficient Strategies for Reducing Sampling Error in Quantum Krylov Subspace Diagonalization}

\author{Gwonhak Lee}
\affiliation{SKKU Advanced Institute of Nanotechnology (SAINT), Sungkyunkwan University, Suwon 16419, Korea}

\author{Seonghoon Choi}%
\affiliation{Chemical Physics Theory Group, Department of Chemistry, University of Toronto, Toronto, Ontario M5S 3H6, Canada}
\affiliation{Department of Physical and Environmental Sciences, University of Toronto Scarborough, Toronto, Ontario M1C 1A4, Canada}

\author{Joonsuk Huh}
\email{joonsukhuh@gmail.com}
\affiliation{Department of Chemistry, Yonsei University,  Seoul 03722, South Korea.}%

\author{Artur F. Izmaylov}%
\email{artur.izmaylov@utoronto.ca}
\affiliation{Chemical Physics Theory Group, Department of Chemistry, University of Toronto, Toronto, Ontario M5S 3H6, Canada}
\affiliation{Department of Physical and Environmental Sciences, University of Toronto Scarborough, Toronto, Ontario M1C 1A4, Canada }

\date{\today}
\begin{abstract}
    Within the realm of early fault-tolerant quantum computing (EFTQC), quantum Krylov subspace diagonalization (QKSD) has emerged as a promising quantum algorithm for the approximate Hamiltonian diagonalization via projection onto the quantum Krylov subspace.
    However, the algorithm often requires solving an ill-conditioned generalized eigenvalue problem (GEVP) involving erroneous matrix pairs, which can significantly distort the solution.
    Since EFTQC assumes limited-scale error correction, finite sampling error becomes a dominant source of error in these matrices.
    This work focuses on quantifying sampling errors during the measurement of matrix element in the projected Hamiltonian examining two measurement approaches based on the Hamiltonian decompositions: the linear combination of unitaries and diagonalizable fragments.
    To reduce sampling error within a fixed budget of quantum circuit repetitions, we propose two measurement strategies: the shifting technique and coefficient splitting. The shifting technique eliminates redundant Hamiltonian components that annihilate either the bra or ket states, while coefficient splitting optimizes the measurement of common terms across different circuits.
    Numerical experiments with electronic structures of small molecules demonstrate the effectiveness of these strategies, reducing sampling costs by a factor of 20-500.
\end{abstract}
\maketitle


\section{\label{sec:Introduction}Introduction}

Recent advancements in quantum computing devices, particularly in terms of the scale and coherence time \cite{Wintersperger2023-wc, Debnath2016, 10.1063/1.4966970, PhysRevLett.117.210502, Qiang2018, PhysRevLett.109.060501, Wendin_2017}, have significantly heightened expectations for their ability to perform efficient quantum simulations.
These advancements promise to deepen our understanding of many-body quantum systems, such as electronic structure in chemical systems \cite{doi:10.1021/acs.chemrev.8b00803, RevModPhys.92.015003, doi:10.1021/acs.chemrev.9b00829}.
This anticipation is driven by the expected stability, controllability, and scalability of universal quantum computers being developed across various platforms, including ion traps \cite{Debnath2016, 10.1063/1.4966970}, photons \cite{PhysRevLett.117.210502, Qiang2018}, and superconductors \cite{PhysRevLett.109.060501, Wendin_2017}.

Currently, the field is progressing through the era of noisy and intermediate-scale quantum computers (NISQ) \cite{Preskill2018quantumcomputingin, RevModPhys.94.015004}.
This phase marks a regime of quantum computation which is hard to be simulated using classical computers, while quantum error correction is absent due to the limited scalability, inherent noises, and decoherence of current devices.
Within this context, the variational quantum eigensolver (VQE) has been primarily discussed as an algorithm for quantum simulation \cite{Cerezo2021, TILLY20221}.
Based on the variational principle, VQE employs quantum-classical hybrid optimization of a parameterized ansatz, implementable within a shallow quantum circuit to approximate specific eigenstates of the target system.
However, the expected quantum advantage\textemdash based on the classical hardness of simulating the ansatz\textemdash is negated by errors in estimating cost function for each optimization step, particularly associated with barren plateau problem \cite{McClean2018, Cerezo2021_BP, Wang2021, Cerezo_2021_BP_derivative, Arrasmith2021effectofbarren}.
Furthermore, the absence of the error correction results in the errors accumulating significantly, thus limiting the scalable quantum advantage in VQE.

This naturally shifts our attention towards early fault-tolerant quantum computing (EFTQC) as a viable next step beyond NISQ era.
The feasibility of EFTQC is further supported by decreasing hardware error rates that are approaching the threshold for the error correction \cite{Xue2022-kr, Blume-Kohout2017-oj, Postler2022-iv} and an emergence of a small-scale demonstration of logical qubits \cite{Bluvstein2024}.
EFTQC is introduced within a framework of scale-limited quantum error correction, where the error rate for logical qubits increases with the size of the quantum circuit \cite{katabarwa2023early}.
Consequently, unlike fully fault-tolerant quantum computing, EFTQC cannot arbitrarily scale the number of logical qubits or the use of non-Clifford operations, thereby limiting the practical implementation of quantum phase estimation \cite{PhysRevLett.83.5162}.
To address these limitations, EFTQC algorithms typically aim to compromise between the circuit size and the number of repetitions.
Quantum phase estimation requires $M=O(|\gamma_0|^{-2})$ repetitions of a circuit with real-time propagators ($e^{-i\hat{H}t_k}$) with the total propagation time, $\sum_kt_k=O(\epsilon_{\mathrm{alg}}^{-1})$, where $|\gamma_0|$ and $\epsilon_{\mathrm{alg}}$ denote the initial overlap and algorithm accuracy, respectively \cite{PhysRevLett.83.5162}.
Here, the real-time propagator is usually approximately implemented by Suzuki-Trotterization \cite{fractal_decomposition, Suzuki_book}.
In contrast, EFTQC algorithms use shorter propagators, $t<O(\epsilon_{\mathrm{alg}}^{-1})$, but increase the number of repetitions, $M>O(|\gamma_0|^{-2})$ \cite{Somma_2019, Zhang2022computingground, PRXQuantum.4.020331, Ding2023simultaneous, parrish2019quantum, Stair2020-hq}.
An alternative approach within EFTQC utilizes the block encoding scheme \cite{Kirby2023exactefficient}, which, while exact, demands significantly more resources than Trotterization with a minimal Trotter steps.
Although reducing Trotterization errors necessitates more Trotter steps, thus approaching the resource demands of block encoding, the limitations of near-future hardware render small-scale Trotterization a more feasible option.

Within the domain of EFTQC algorithms, quantum Krylov subspace diagonalization (QKSD) is being explored as a promising candidate for quantum simulation algorithm \cite{parrish2019quantum, Stair2020-hq}.
QKSD employs quantum circuits to project the Hamiltonian onto the Krylov subspace, a reduced-dimensional space that is classically solvable.
This approach is potentially feasible because the extremal eigenvalues of the projected matrix converge exponentially fast to those of the original Hamiltonian, provided the projection remains unperturbed and the overlap between the eigenstate and the initial state is large \cite{Theory_QKSD}.
However, this advantage is counterbalanced by the challenge associated with ill-conditioning of the eigenvalue problem, where perturbations in the projected matrix can significantly distort the accuracy of the approximated eigenvalues \cite{Theory_QKSD, Lee2024samplingerror}.
These perturbations mainly arise from imperfect error correction, Trotterization error, and finite sampling error.
As EFTQC stabilizes and expands, the first two factors can be suppressed.
However, despite the discussion of the measurement problem in QKSD \cite{Lee2024samplingerror}, strategies to tackle this problem have not been suggested.

QKSD involves measuring the matrix elements, $\bm{H}_{kl}=\braket{\phi_k |\hat{H}|\phi_l}$, across a finite basis $\{\ket{\phi_k}=e^{-i\hat{H}k\Delta_t}\ket{\phi_0}\}$ that spans the Krylov subspace with a reference state $\ket{\phi_0}$.
While other bases have been proposed, such as $\ket{\phi_k}=\hat{H}^k\ket{\phi_0}$\cite{PRXQuantum.2.010333} and $\ket{\phi_k}=e^{-\hat{H}k\beta}\ket{\phi_0}$\cite{Motta2020-oi}, we focus on QKSD with the real-time evolution operator due to its simplicity and practical viability for EFTQC.
Consequently, a primary objective in this scenario is to minimize the sampling error when measuring the matrix elements.
Although this specific measurement problem has not been tackled, several measurement strategies have been proposed for the standard expectation values \cite{Crawford2021efficientquantum, Yen2023, doi:10.1021/acs.jctc.3c00218, doi:10.1021/acs.jctc.0c00008, doi:10.1021/acs.jctc.2c00837, Choi2023fluidfermionic, PRXQuantum.2.040320, 10.1063/1.5141458}.
In general, the direct measurement of $\braket{\psi|\hat{H}|\psi}$ with a single circuit is impractical, as measurement is constrained to the Pauli-Z basis, which requires the implementation of a unitary operator that fully diagonalizes $\hat{H}$.
As an alternative, $\hat{H}$ is decomposed into a linear combination of fragmented Hamiltonians, which can be efficiently diagonalized with implementable unitaries, and the results for each fragment are aggregated \cite{doi:10.1021/acs.jctc.0c00008, 10.1063/1.5141458, Crawford2021efficientquantum}.
The goal then becomes to minimize the sampling error by optimally fragmenting $\hat{H}$ and allocating the number of circuit repetitions among these fragments.
This optimization can be formulated as a combinatorial problem, akin to NP-hard clique covering problems \cite{doi:10.1021/acs.jctc.0c00008, 10.1063/1.5141458}, for which a heuristic solution has been developed \cite{Crawford2021efficientquantum}.
Additionally, this issue has been expanded into a continuous optimization problem that considers approximated covariances between fragments \cite{Choi2023fluidfermionic, Yen2023}.
On the other hand, randomized measurement strategies, known as classical shadow techniques, have been developed \cite{Huang2020, hadfield2020lbcs, PhysRevLett.127.030503}.
While these methods are superior when measuring multiple expectation values simultaneously, deterministic methods generally outperform them when measuring a single expectation value $\braket{\hat{H}}$ \cite{Yen2023, doi:10.1021/acs.jctc.3c00218, Nakamura_2024}.
Since our case involves measuring $\braket{\phi_k|\hat{H}|\phi_l}$, we focus on deterministic measurement strategies in this work.

These developments resolve the measurement problems associated with standard expectation values.
In this paper, we aim to adapt these methods to the QKSD framework.
To accomplish this, we analyze the measurement problem of the matrix elements with two decomposition scenarios for $\hat{H}$: linear combination of unitaries (LCU) and fragmented Hamiltonians (FH).
Specifically, we focus on quantifying and mitigating the sampling errors for these scenarios, applying strategies originally designed for standard expectation values to enhance measurement accuracy in QKSD.
Notably, the strategies that we propose can be applied to the general measurement of transitional amplitudes, $\braket{\phi_k|\hat{H}|\phi_l}$, which can be utilized for the design of algorithms beyond QKSD.

The organization of the paper is as follows.
First, a brief preliminary of the QKSD is provided in Section \ref{sec:QKSD}, followed by the analyses of the two measurement methods for the matrix elements and the associated errors in Section \ref{sec:Measurement of QKSD Matrix Elements}.
Section \ref{sec:Sampling Cost Reduction} demonstrates how conventional methods \cite{Crawford2021efficientquantum, Yen2023, Loaiza_BLISS, Loaiza_2023} are converted to reduce the sampling errors for QKSD, highlighting that a method initially devised for reducing the simulation cost of LCU \cite{Loaiza_BLISS, Loaiza_2023} is transferable to the measurement problem.
Finally, we numerically validate these reduction methods by solving the electronic structure problems of small molecules as case studies in Section \ref{sec:numerical_analysis}.

\section{QKSD \label{sec:QKSD}}
Before considering the measurement problem in QKSD, this section reviews the QKSD method for estimating the spectrum of a Hamiltonian $\hat{H}$, as originally introduced in \cite{parrish2019quantum, Stair2020-hq}.

QKSD estimates approximated eigenstates of a Hamiltonian $\hat{H}$ with the following ansatz:
\begin{equation}\label{eq:QKSD_ansatz}
    \ket{\psi(\bm{w})} := \frac{1}{\mathcal{N}}\sum_{k=0}^{n-1} w_k\hat{B}^k\ket{\phi_0},
\end{equation}
where $\mathcal{N}$ is the normalization factor, and $\bm{w} = (w_0, \cdots, w_{n-1})\in \mathbb{C}^{n}$, and $n$ is the Krylov order.
This ansatz fully covers vectors in the Krylov subspace, $\mathcal{K}_n=\mathrm{span}(\{\ket{\phi_0}, \cdots \ket{\phi_{n-1}}\})$, where $\ket{\phi_k} = \hat{B}^{k}\ket{\phi_0}$ is defined by a reference state $\ket{\phi_0}$ and the base operator,
\begin{equation}
    \hat{B}=e^{-i\Delta_t \hat{H}}.
\end{equation}
This exponential function is approximated by the Trotterization.
There are other choices for $\hat{B}$, such as $\hat{H}$, which is analogous to the classical Krylov method \cite{PRXQuantum.2.010333}, and imaginary time evolution ($e^{-\beta \hat{H}}$) \cite{Motta2020-oi}.
Although we only focus on the real-time evolution operator, which is widely discussed, the methods that will be described in the Section \ref{sec:Sampling Cost Reduction} can be expanded to the other choices of $\hat{B}$ because they are approximated by linear combinations or products of real-time propagators.

Using the ansatz (Eq.\eqref{eq:QKSD_ansatz}) in the variational principle with $\bm{w}$ as optimized parameters leads to the following generalized eigenvalue problem (GEVP)
\begin{equation}\label{eq:gevp}
    \bm{H}\bm{w} = \bm{S}\bm{w}E^{(n)},
\end{equation}
where $E^{(n)}$ is an approximate eigenvalue, and the $n\times n$ Hamiltonian matrix $\bm{H}$ and overlap matrix $\bm{S}$ are defined as
\begin{align}
    \bm{H}_{kl} = \bra{\phi_k}\hat{H}\ket{\phi_l} =&\braket{\phi_0|\hat{B}^{\dagger k}\hat{H}\hat{B}^{l}|\phi_0}=\bm{H}_{0,l-k} \label{eq:QKSD_H}\\
    \bm{S}_{kl} = \braket{\phi_k|\phi_l}=&\braket{\phi_0|\hat{B}^{\dagger k}\hat{B}^{l}|\phi_0}=\bm{S}_{0,l-k}\label{eq:QKSD_S}.
\end{align}
These matrices are obtained by quantum algorithms using EFTQC circuits and measurements.
Note that $\bm{S}$ has the structure of a Toeplitz matrix, $\bm{S}_{k,l}=\bm{S}_{0,l-k}$, because $\hat{B}$ is unitary.
Furthermore, since $[\hat{H}, \hat{B}]=0$, $\bm{H}$ becomes a Toeplitz matrix as well.
Therefore, rather than $n^2$, only $n$ elements are required to construct each matrices.

For systems with large energy gaps, the lowest solution of Eq.\eqref{eq:gevp} converges to the ground state energy of $\hat{H}$ exponentially fast with $n$ \cite{Theory_QKSD}.
However, the GEVP in Eq.\eqref{eq:gevp} can become ill-conditioned for larger $n$'s, which makes QKSD sensitive to noise in matrix elements of $\bm{H}$ and $\bm{S}$.

\section{\label{sec:Measurement of QKSD Matrix Elements}Measurement of QKSD Matrix Elements}

In the quantum subroutine that estimates the elements of $\bm{H}$ and $\bm{S}$, the quantum uncertainty predominantly induces the matrix perturbation.
This perturbation, coupled with the ill-conditioned GEVP (Eq.\eqref{eq:gevp}), may introduce a significant error in the solution.
In this section, we develop methods for estimating Hamiltonian matrix elements,
\begin{equation}\label{eq:matrix_element}
    \bm{H}_{0k}=\braket{\phi_0 |\hat{H}|\phi_k},
\end{equation}
and the analysis of the associated sampling error.

Since, rather than the matrix element, only the measurements of standard expectation values of $\braket{\Phi|\hat{O}_j|\Phi}$ with easily diagonalizable $\hat{O}_j$'s are possible at the circuit level, it is necessary to express Eq.\eqref{eq:matrix_element} in terms of $\braket{\Phi|\hat{O}_j|\Phi}$ using certain states $\ket{\Phi}$ and simple operators $\hat{O}_j$.
To translate Eq.\eqref{eq:matrix_element} into standard expectations, we consider two approaches: the Hadamard and the extended swap tests.
In both approaches, the problem is addressed by partitioning the Hamiltonian into diagonalizable Hermitian or implementable unitary operators.
In the Hadamard test, the Hamiltonian is presented as an LCU: $\hat{H}=\sum_j\beta_j\hat{U}_j$, where each unitary $\hat{U}_j$ is implemented to estimate the overlap between $\ket{\phi_0}$ and $\hat{U}_j\ket{\phi_k}$ (Fig.~\ref{fig:ht_ckt}).
In the extended swap test, $\hat{H}$ is decomposed as a sum of fragments: $\hat{H}=\sum_j\hat{V}^\dagger_j \hat{D}_j \hat{V}_j$, where $\hat{D}_j$ is diagonal and the corresponding diagonalizing unitary $\hat{V}_j$ can be implemented efficiently.
Then, the extended swap test is conducted with the circuits of Hadamard tests which estimate $\braket{\phi_0|\phi_k}$, while including the measurement of system qubits with the basis represented by $\hat{V}_j$ (Fig.~\ref{fig:ex_swap_ckt}).
Note that in the extended swap test, the overlap matrix elements can be measured simultaneously by measuring the ancilla qubit, which corresponds to a Hadamard test circuit that measures $S_{0,k}=\braket{\phi_0|e^{-ik\Delta_t\hat{H}}|\phi_0}$.
This simultaneous measurement capability is not available in the Hadamard test when measuring the Hamiltonian in its LCU form.

This work mainly focuses on analyzing and improving the sampling error determined by the decomposition of the observable, which is not required for the case of the overlap matrix element, $\bm{S}_{0k}=\braket{\phi_0|\hat{\mathbb{I}}|\phi_k}$.
In a previous work \cite{Lee2024samplingerror}, it was shown that the sampling variance of $\bm{S}_{0k}$ are determined identically across the test algorithms as:
\begin{equation}
\begin{split}
    \mathrm{Var}&[\bm{S}_{0k};m_R, m_I]=\\
    &\frac{1}{M}\left(\frac{1-\mathrm{Re}[\bm{S}_{0k}]^2}{m_R}+\frac{1-\mathrm{Im}[\bm{S}_{0k}]^2}{m_I}\right),
\end{split}
\end{equation}
where $m_{R}$ and $m_{I}$ are the fraction of total shots, $M$, allocated to measure real and imaginary parts of $\bm{S}_{0k}$, respectively ($m_R+m_I = 1$).
The ideal allocation that minimizing the variance is $m_R^{(\mathrm{opt})}\propto (1-\mathrm{Re}[\bm{S}_{0k}]^2)^{1/2}$ and $m_I^{(\mathrm{opt})}\propto (1-\mathrm{Im}[\bm{S}_{0k}]^2)^{1/2}$, which is, however, infeasible to be estimated before measuring $\bm{S}_{0k}$.
If we take Haar averaging of the states, $\ket{\phi_0}$ and $\ket{\phi_k}$, $m_R=m_I=1/2$ is achieved from Lemma \ref{lemm:haar_reim_square_independent} in Appendix \ref{sec:appendix_asymptotic_proof}.
The corresponding variance is related to the amplitude of the matrix element as
\begin{equation}
    \mathrm{Var}[\bm{S}_{0k};1/2,1/2] = \frac{2}{M}\left(2 - |\bm{S}_{0k}|^2\right).
\end{equation}

In the rest of this section, we quantify the sampling error associated with the Hamiltonian matrix elements and examine how the decomposition of $\hat{H}$ affects this error.

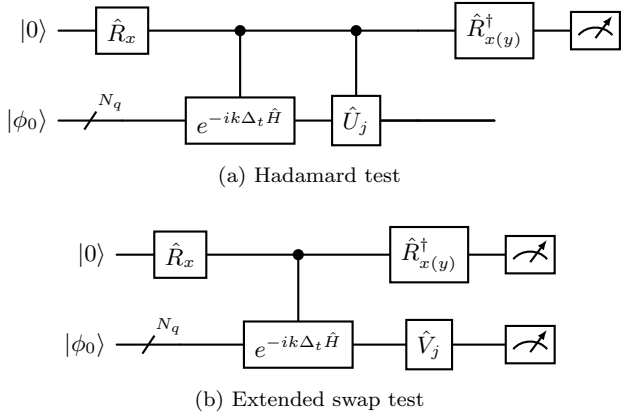
\begin{figure}[t]
    \centering
    \subfloat[Hadamard test]{
        \label{fig:ht_ckt}
        \begin{quantikz}
        \lstick{$\ket{0}$}      & \gate{\hat{R}_x}       & \ctrl{1} & \ctrl{1} & 
        & \gate{\hat{R}_{x(y)}^\dagger}  & \meter{} \\
        \lstick{$\ket{\phi_0}$} & \qwbundle{N_q} & \gate{e^{-ik\Delta_t\hat{H}}} & \gate{\hat{U}_j} & \qw & \qw
        \end{quantikz}
    }
    \\
    \subfloat[Extended swap test]{
        \label{fig:ex_swap_ckt}
        \begin{quantikz}
        \lstick{$\ket{0}$}      & \gate{\hat{R}_x} & \ctrl{1} 
        & \gate{\hat{R}_{x(y)}^\dagger}  & \meter{} \\
        \lstick{$\ket{\phi_0}$} & \qwbundle{N_q} & \gate{e^{-ik\Delta_t\hat{H}}}      &\gate{\hat{V}_j} & \meter{}
        \end{quantikz}
    }
    \caption{
    Circuit diagram for (a) the Hadamard test and (b) the extended swap test to estimate the $j$-th fragment of the Hamiltonian matrix element, $\braket{\phi_0|\hat{H}|\phi_k}$(see Eqs.\eqref{eq:LCU_measurement} and \eqref{eq:EST_measurement}.)
    Here, $\hat{R}_{x(y)}$ operator rotates $\hat{\sigma}_z$ basis into $\hat{\sigma}_x$ ($\hat{\sigma}_y$) basis and is adopted to estimate the real (imaginary) part of the amplitude for the second operator applied to the ancilla qubit.
    }
\end{figure}

\subsection{Hadamard Test}
An LCU decomposition of $\hat{H}$ is
\begin{equation}\label{eq:LCU}
    \hat{H}=\sum_{j=1}^{N_\beta}\beta_j \hat{U}_j,
\end{equation}
where $\hat{U}_j$ is unitary, and $\beta_j$ is a real and positive coefficient.
Such decomposition was originally motivated in the context of Hamiltonian simulation problem \cite{LCU_childs} and then expanded to the measurement problem \cite{LCU_Izmaylov, LCU_Love}.
The simulation cost in LCU based approach scales linearly with the L1 norm of the coefficients, $\|\bm{\beta}\|_1=\sum_j|\beta_j|$.
In the Hadamard test, we are going to show that the sampling cost to estimate Eq.\eqref{eq:matrix_element} within a certain level of accuracy is proportional to $\|\bm{\beta}\|_1^2$.

The matrix element in Eq.\eqref{eq:matrix_element} can be viewed as the weighted sum of overlaps between $\ket{\phi_0}$ and $\hat{U}_j \ket{\phi_k}$:
\begin{equation}
    \braket{\phi_0|\hat{H}|\phi_k}=\sum_{j=1}^{N_\beta}\beta_j \braket{\phi_0|\hat{U}_j|\phi_k}.
\end{equation}
Figure \ref{fig:ht_ckt} depicts the circuit that estimates $\braket{\phi_0|\hat{U}_j|\phi_k}$.
To derive the sampling cost, Eq.\eqref{eq:matrix_element} is translated to standard expectations as
\begin{equation}\label{eq:LCU_measurement}
    \braket{\phi_0|\hat{H}|\phi_k} = \sum_{j=1}^{N_\beta}\beta_j\braket{\Phi_{0k;j}|(\hat{\sigma}_x + i\hat{\sigma}_y)\otimes \hat{\mathbb{I}}|\Phi_{0k;j}},
\end{equation}
where
\begin{equation}\label{eq:LCU_state}
    \ket{\Phi_{0k;j}} := \frac{1}{\sqrt{2}}
    \left(\ket{0}\ket{\phi_0}+\ket{1}\hat{U}_j\ket{\phi_k}\right)
\end{equation}
is prepared with an additional qubit and the conditional evolution unitary.
The additional qubit is measured in $\hat{\sigma}_x$ and $\hat{\sigma}_y$ bases, corresponding to the real and imaginary parts of $\braket{\phi_0|\hat{U}_j|\phi_k}$, respectively.
Thus, $2N_\beta$ independent measurements of $\hat{O}_{R}=\hat{\sigma}_x\otimes \hat{\mathbb{I}}$ and $\hat{O}_{I}= \hat{\sigma}_y\otimes \hat{\mathbb{I}}$ with a set of states $\{\ket{\Phi_{0k;j}}\}_{j=1}^{N_\beta}$ complete the total estimation of a matrix element in Eq.\eqref{eq:LCU_measurement}.

For measuring the expectation value of Eq.\eqref{eq:LCU_measurement}, the variance is:
\begin{equation}\label{eq:LCU_variance}
\begin{split}
    \mathrm{Var}^{\mathrm{(LCU)}}&[\bm{H}_{0k};\bm{m}]=\\
    &\frac{1}{M}\left(\sum_{j=1}^{N_\beta}\sum_{X\in \{R, I\}}\frac{\beta_j^2}{m_{jX}}\mathrm{Var}[\hat{O}_X]_{\Phi_{0k;j}}\right),    
\end{split}
\end{equation}
where
\begin{equation}\label{eq:LCU_partial_variance}
\begin{split}
    \mathrm{Var}[\hat{O}_{X}]_{\Phi_{0k;j}}=&\braket{\hat{O}_X^2}_{\Phi_{0k;j}}-\braket{\hat{O}_X}_{\Phi_{0k;j}}^2 \\
    =&\begin{cases}
        1 - \mathrm{Re}[\braket{\phi_0|\hat{U}_j|\phi_k}]^2, & \ X = R\\
        1 - \mathrm{Im}[\braket{\phi_0|\hat{U}_j|\phi_k}]^2, & \ X = I.
    \end{cases}
\end{split}
\end{equation}
Also, $M$ denotes the total number of shots to measure $\bm{H}_{0k}$, and $m_{jX}$ is the fraction of shots for $\braket{\Phi_{0k;j}|\hat{O}_X|\Phi_{0k;j}}$ ($\sum_{jX}m_{jX}=1$.)

Given LCU decomposition, finding the optimal shot allocation $\bm{m}$ to minimize Eq.\eqref{eq:LCU_variance} is a convex problem, which is analytically solved by 
\begin{equation}\label{eq:LCU_optimal_allocation}
    m^{\mathrm{(opt)}}_{jX}\propto\beta_j {\mathrm{Var}[\hat{O}_X]_{\Phi_{0k;j}}^{1/2}}.
\end{equation}
However, since $\mathrm{Var}[\hat{O}_X]_{\Phi_{0k;j}}$ are not known in advance,
they are estimated by taking Haar-averaging over the states $\ket{\phi_0}$ and $\ket{\phi_k}$.
This results in a sub-optimal allocation instead, $m^{\mathrm{(subopt)}}_{jX}\propto\beta_j$, which is independent of the unknown variances.
Furthermore, such allocation leads to the averaged variance as shown below:
\begin{equation}
\begin{split}\label{eq:asymptotic_LCU_variance}
    V^{(\mathrm{LCU})} := &\mathbb{E}_{\phi_0, \phi_k}\left[\mathrm{Var}^{(\mathrm{LCU})}\left[\bm{H}_{0k};\bm{m}^{(\mathrm{subopt})}\right]\right]\\
= &\frac{2\|\bm{\beta}\|_1^2}{M}\left(2 - \frac{1}{d}\right),
\end{split}
\end{equation}
where $d$ is the dimension of the Hilbert space, which is generally exponentially large.
The proof of Eq.\eqref{eq:asymptotic_LCU_variance} is provided in Appendix \ref{sec:appendix_asymptotic_proof}.
The number of total shots more than
\begin{equation}\label{eq:LCU_samp_cost}
    M\approx\frac{4\|\bm{\beta}\|_1^2}{\epsilon^2}
\end{equation}
is required to make the total uncertainty less than $\epsilon$ with a high probability, which shows that the sampling cost is proportional to $\|\bm{\beta}\|_1^2$.
This result indicates that an LCU decomposition of $\hat{H}$ with lower $\|\bm{\beta}\|_1$ can potentially reduce the number of samplings, while maintaining accuracy.

\subsection{Extended Swap Test}
The extended swap test \cite{parrish2019quantum} estimates the matrix element through the measurement of
\begin{equation}\label{eq:EST_measurement}
    \braket{\phi_0|\hat{H}|\phi_k}=\braket{\Phi_{0k}|\left(\hat{\sigma}_x+i\hat{\sigma}_y\right)\otimes \hat{H}|\Phi_{0k}},
\end{equation}
where
\begin{equation*}
    \ket{\Phi_{0k}}:=\frac{1}{\sqrt{2}}\left(\ket{0}\ket{\phi_0}+\ket{1}\ket{\phi_k}\right).
\end{equation*}

However, the direct measurement of Eq.\eqref{eq:EST_measurement} is possible only when $\hat{H}$ is an Ising form. 
In general, $\hat{H}$ can be expressed as a sum of diagonalizable fragment Hamiltonians:
\begin{equation}\label{eq:fragmented_hamiltonian}
    \hat{H} = \sum_{j=1}^{N_\gamma} \hat{H}_j = \sum_{j=1}^{N_\gamma} \hat{V}_j^\dagger\hat{D}_j\hat{V}_j,
\end{equation}
where $\hat{V}_{j}$ can be efficiently implemented to diagonalize $\hat{H}_j$ onto the computational basis, yielding an Ising Hamiltonian $\hat{D}_j$.
Then, the measurement can be performed for each $\hat{H}_j$.
For example, if $\hat{H}_j$ is composed of mutually commuting Pauli operators, the unitary $\hat{V}_j$ can be efficiently determined and implemented as a Clifford circuit with the gate count of $O(N_q^2)$, where $N_q$ is the number of qubits \cite{zheng2018depth, gokhale2019minimizing, Crawford2021efficientquantum}.

After substituting $\hat{H}$ in Eq.\eqref{eq:EST_measurement} by Eq.\eqref{eq:fragmented_hamiltonian}, the total variance of the estimator is
\begin{equation}\label{eq:FH_variance}
    \mathrm{Var}^{\mathrm{(FH)}}[\bm{H}_{0k};\bm{m}]=\frac{1}{M}\left(\sum_{j=1}^{N_\gamma}\sum_{X\in\{R, I\}}\frac{\mathrm{Var}[\hat{O}_{jX}]_{\Phi_{0k}}}{m_{jX}}\right),
\end{equation}
where $\hat{O}_{j,R(I)}=\hat{\sigma}_{x(y)}\otimes \hat{H}_j$, and the quantum variance of $\hat{O}_j$ is
\begin{equation}\label{eq:FH_partial_variance}
    \mathrm{Var}[\hat{O}_{jX}]_{\Phi_{0k}}=\begin{cases}
        \braket{\hat{O}_{jR}^2}_{\Phi_{0k}} - \mathrm{Re}[\braket{\phi_0|\hat{H}_j|\phi_k}]^2
         & X=R, \\
        \braket{\hat{O}_{jI}^2}_{\Phi_{0k}} - \mathrm{Im}[\braket{\phi_0|\hat{H}_j|\phi_k}]^2
         & X=I,
    \end{cases}
\end{equation}
with 
\begin{equation*}
\braket{\hat{O}_{jR}^2}_{\Phi_{0k}}=\braket{\hat{O}_{jI}^2}_{\Phi_{0k}}=\frac{1}{2}\left(\braket{\phi_0|\hat{H}_j^2|\phi_0}+\braket{\phi_k|\hat{H}_j^2|\phi_k}\right).    
\end{equation*}

As in LCU decomposition, the optimal shot allocation, $m_{jX}^{\mathrm{(opt)}}\propto \mathrm{Var}[\hat{O}_{jX}]_{\Phi_{0k;j}}^{1/2}$, is hard to estimate in advance.
In order to find a sub-optimal allocation, $\bm{m}^{\mathrm{(subopt)}}$, we perform Haar averaging on the variance over the states, $\ket{\phi_0}$ and $\ket{\phi_k}$.
According to the result in Appendix \ref{sec:appendix_asymptotic_proof}, this averaging produces $m^{\mathrm{(subopt)}}_{jR(I)}\propto (\mathrm{Tr}[\hat{H}_j^2]/d)^{1/2}$.
Furthermore, for the case of Pauli decomposition, the value of $(\mathrm{Tr}[\hat{H}_j^2]/d)^{1/2}$ is efficiently determined as the L2 norm of Pauli coefficients in $\hat{H}_j$, as elaborated in Appendix \ref{sec:appendix_pauli_grouping}.

The total averaged variance is \begin{equation}\label{eq:asymptotic_FH_variance}
\begin{split}
    V^{\mathrm{(FH)}}:=&\mathbb{E}_{\phi_0, \phi_k}\left[\mathrm{Var}^{(\mathrm{FH})}\left[\bm{H}_{0k};\bm{m}^{\mathrm{(subopt)}}\right]\right]\\
    =&\frac{2\|\bm{\gamma}\|_1^2}{M}\left(2-\frac{1}{d}\right),
\end{split}
\end{equation}
where $\|\bm{\gamma}\|_1:=\sum_{j}(\mathrm{Tr}[\hat{H}_j^2]/d)^{1/2}$.
The corresponding measurement cost,
\begin{equation}\label{eq:FH_samp_cost}
    M \approx \frac{4\|\bm{\gamma}\|_1^2}{\epsilon^2},
\end{equation}
scales quadratically with the $\|\bm{\gamma}\|_1$, which plays the same role as $\|\bm{\beta}\|_1$ in the LCU case.

We can find the similarity between the variances of LCU and FH, (see Eqs.\eqref{eq:asymptotic_LCU_variance} and \eqref{eq:asymptotic_FH_variance}).
Let us write the decompositions for both cases (Eqs.\eqref{eq:LCU} and \eqref{eq:fragmented_hamiltonian}) as 
\begin{equation}
    \hat{H} = \sum_{j=1}^{N_\zeta}\hat{N}_j,    
\end{equation}
where $\hat{N}_{j}\in \{\beta_j \hat{U}_j, \hat{H}_j\}$.
Then, the partial variances, Eqs.\eqref{eq:LCU_partial_variance} (multiplied by $\beta_j$) and \eqref{eq:FH_partial_variance}, are generalized into a single form:
\begin{equation}\label{eq:normal_partial_variance}
\begin{split}
    \mathrm{Var}[X_{0k;j}]=\frac{1}{2}\left(\braket{\phi_0|\hat{N}_j^\dagger \hat{N}_j|\phi_0}+\braket{\phi_k|\hat{N}_j^\dagger \hat{N}_j|\phi_k}\right)\\- \mathbb{E}[X_{0k;j}]^2,
\end{split}
\end{equation}
where $X_{0k;j}\in \{R_{0k;j}, I_{0k;j}\}$ is an estimator for the real or imaginary part of $\braket{\phi_0|\hat{N}_j|\phi_k}$.

Moreover, the decomposition norms, $\|\bm{\beta}\|_1$ and $\|\bm{\gamma}\|_1$ in Eqs.\eqref{eq:asymptotic_LCU_variance} and \eqref{eq:asymptotic_FH_variance}, are regarded as:
\begin{equation}
    \|\bm{\zeta}\|_1:=\frac{1}{\sqrt{d}}\sum_j \mathrm{Tr}[\hat{N}_j^\dagger \hat{N}_j]^{1/2}.
\end{equation}
Correspondingly, the sampling cost is given by:
\begin{equation}\label{eq:zeta_cost}
    M\approx \frac{4\|\bm{\zeta}\|_1^2}{\epsilon^2}.
\end{equation}
Thus, regardless of the test algorithms, $\|\bm{\zeta}\|_1$ serves as a metric of decomposition assessing the finite sampling error, akin to the measurements of standard expectation values \cite{LCU_Love, Crawford2021efficientquantum}.



\section{\label{sec:Sampling Cost Reduction}Sampling Cost Reduction}
In this section, we propose techniques to reduce the sampling cost discussed in the previous section.
This is done by adapting the cost reduction techniques for the measurement of standard expectation, $\braket{\phi|\hat{H}|\phi}$, \cite{Crawford2021efficientquantum, Yen2023, Choi2023fluidfermionic} to the measurement of matrix elements.
Such adaptation is simply done by replacing the standard variance, $\braket{\hat{O}_j^2}-\braket{\hat{O}_j}^2$, by the variance for the matrix elements represented in Eqs.\eqref{eq:LCU_partial_variance} and \eqref{eq:FH_partial_variance}.
We propose a method to optimize the decomposition, $\{\hat{U}_j\}$ or $\{\hat{H}_j\}$, to achieve smaller variance analogous to the approach in previous work \cite{Yen2023}.
Furthermore, the dependence of the sampling cost on $\|\bm{\zeta}\|_1$, as shown in Eq.\eqref{eq:zeta_cost}, allows us to use methods that reduce $\|\bm{\zeta}\|_1$.


\subsection{\label{subsec:Shifting Technique}Shifting Technique}
Here, we introduce a technique to reduce the norm $\|\bm{\zeta}\|_1$, and consequently, lower expected sampling costs, by shifting the Hamiltonian.
Before developing the technique, let's clarify the notation for the norm, $\|\bm{\zeta}_{A}(\hat{H})\|_1$, which indicates the norm of the decomposition achieved on $\hat{H}$ using a specific deterministic algorithm, $A$.
This clarification is crucial because the decomposition applied to $\hat{H}$ is not unique without specifying the algorithm $A$.
As an example of $A$, a greedy algorithm-based $\texttt{SORTED INSERTION}$ heuristically finds a decomposition by Pauli operators that yields a relatively small norm \cite{Crawford2021efficientquantum}.

The shifting technique involves finding an operator $\hat{T}$ that shifts the Hamiltonian and minimizes the norm of the shifted decomposition:
\begin{equation}\label{eq:min_norm}
    \min_{\bm{\tau}} \|\bm{\zeta}_A(\hat{H} - \hat{T}(\bm{\tau}))\|_1,
\end{equation}
where $A$ is a fixed polynomial time algorithm performing a decomposition, and the Hermitian operator $\hat{T}$ is parameterized by $\bm{\tau}$, enabling the use of classical optimization algorithms.
Additionally, we impose a constraint on $\hat{T}(\bm{\tau})$, that is
\begin{equation}\label{eq:shifting_constraint}
    \hat{T}(\bm{\tau})\ket{\phi_0} = t(\bm{\tau})\ket{\phi_0},
\end{equation}
for a known factor $t(\bm{\tau})\in\mathbb{R}$.
Note that $\hat{T}$ is not required to commute with $\hat{H}$, unlike symmetry operators.
The necessity of this constraint will be presented with the rest of procedure.

After the optimization of Eq.\eqref{eq:min_norm}, we then employ test algorithms explained in Section \ref{sec:Measurement of QKSD Matrix Elements} to estimate the shifted Hamiltonian matrix $\bm{H}-\bm{T}$, consuming reduced cost ($\|\bm{\zeta}_A(\hat{H}-\hat{T})\|_1^2 \le \|\bm{\zeta}_A(\hat{H})\|_1^2$).
Here, $\bm{H}-\bm{T}$ is defined as a Toeplitz matrix satisfying
\begin{equation}\label{eq:shifted_hamiltonian_matrix}
    [\bm{H}-\bm{T}]_{kl} = \braket{\phi_0|(\hat{H}-\hat{T})|\phi_{l-k}}.
\end{equation}
Then, the GEVP with the shifted Hamiltonian matrix is 
\begin{equation}\label{eq:shifted_gevp}
    (\bm{H}-\bm{T})\bm{w}=\bm{S}\bm{w}(E^{(n)}-t),
\end{equation}
because the original matrix element is written in terms of the shifted matrix element as shown below:
\begin{equation*}
\begin{split}
    \bm{H}_{0k} &= \braket{\phi_0|\hat{H}|\phi_k}\\
                                  &=\braket{\phi_0|\hat{H}-(\hat{T}-t)|\phi_k} \\
                                  &=\braket{\phi_0|(\hat{H}-\hat{T})|\phi_k} + t\braket{\phi_0|\phi_k}\\
                                  &=[\bm{H}-\bm{T}]_{0k} + t\bm{S}_{0k},
\end{split}
\end{equation*}
and $\bm{H}-\bm{T}$ is Toeplitz as defined in Eq.\eqref{eq:shifted_hamiltonian_matrix}.
Thus, the constraints of Eqs.\eqref{eq:shifting_constraint} and \eqref{eq:shifted_hamiltonian_matrix} enables the recovery of the solutions of the original GEVP, $E^{(n)}$ from that of shifted one by simply adding $t$.

Here, we give an example of designing and parameterizing $\hat{T}$ dedicated to the electronic structure problem.
In many cases, the reference state, $\ket{\phi_0}$ is chosen as a simple state, such as Hartree Fock (HF) ground state or a configuration state function (CSF) which is a symmetry-adapted state composed of a small number of Slater determinants.
Some or all orbitals in such reference states are separately occupied or unoccupied, which is represented as:
\begin{equation}\label{eq:reference_state}
    \ket{\phi_0} = \left(\bigotimes_{q \in \mathrm{occ}}\ket{1}_q \bigotimes_{q \in \mathrm{virt}}\ket{0}_q\right)\otimes \ket{\phi_0^{(r)}},
\end{equation}
where `$\mathrm{occ}$' and `$\mathrm{virt}$' denote the sets of occupied and virtual orbitals, respectively, and $\ket{\phi_0^{(r)}}$ is possibly entangled and in the remainder system, satisfying $\ket{\phi_0^{(r)}}\bra{\phi_0^{(r)}}=\mathrm{Tr}_{\mathrm{occ},\mathrm{virt}}[\ket{\phi_0}\bra{\phi_0}]$.
In a HF ground state, all orbitals are either occupied or unoccupied while a CSF state may involve some entangled state $\ket{\phi_0^{(r)}}$.
Furthermore, an electronic structure Hamiltonian is represented as
\begin{equation}\label{eq:electronic_structure_hamiltonian}
    \hat{H} = \sum_{p\le q}^{N_{\mathrm{orb}}} h_{pq}\hat{E}_{pq} + \sum_{pqrs}g_{pqrs}(\hat{E}_{pq}\hat{E}_{rs}+\mathrm{h.c.}),
\end{equation}
where $N_{\mathrm{orb}}$ denotes the number of total orbitals and the notations of the symmetric excitation operators $\hat{E}_{rs}=\hat{a}^{\dagger}_r\hat{a}_s +\hat{a}^{\dagger}_s\hat{a}_r$ and the number operators $\hat{n}_q = \hat{a}^{\dagger}_q\hat{a}_q$ are employed.

Then, the shift operator can be designed to cover one- and two-body terms in the Hamiltonian and to satisfy Eq.\eqref{eq:shifting_constraint}, which has the following form:
\begin{equation}\label{eq:shifting_operator_electronic_structure}
\begin{split}
    \hat{T}&(\bm{\tau}^{(1)},\bm{\tau}^{(2)})=\\
    &\sum_{q\in \mathcal{F}}\left(\tau_{q}^{(1)}\hat{n}_q + \sum_{rs\in \mathcal{E}_q}\tau_{qrs}^{(2)}\hat{E}_{rs}(\hat{n}_q - \delta_{q\in \mathrm{occ}})\right).
\end{split}
\end{equation}
Here, the sets of orbital indices, $\mathcal{F}:=\mathrm{occ}\cup \mathrm{virt}$ and $\mathcal{E}_q:=\{(r,s) : r,s\in [N_\mathrm{orb}]\setminus \{q\}, r\le s \}$ are defined, and $\delta_{q\in \mathrm{occ}}$ is an identity if $q \in \mathrm{occ}$, zero otherwise.
The indices $r$ and $s$ range over the entire orbital set except $q$ to make $\hat{T}$ Hermitian and to avoid duplication with the one-body number operator, since $\hat{n}_q^2=\hat{n}_q$.
Note that the two-body terms with $\hat{E}_{rs}\hat{n}_q$ for $q\in\mathrm{virt}$ annihilate $\ket{\phi_0}$, as do $\hat{E}_{rs}(\hat{n}_q - 1)$ for $q\in\mathrm{occ}$.
Therefore, the corresponding shift factor is determined as
\begin{equation}
    t = \sum_{q\in \mathrm{occ}}\tau_q^{(1)}.
\end{equation}

The optimal $\bm{\tau}$ can be found using iterative optimization algorithms like $\mathtt{Powell}$ or $\mathtt{BFGS}$ with the number of parameters of $|\bm{\tau}|=O(N_{\mathrm{orb}}^3)$.
However, the optimization overhead is reduced if we adopt a decomposition algorithm where each term in Eq.\eqref{eq:electronic_structure_hamiltonian} is regarded as a fragment, as detailed in Appendix \ref{sec:appendix_efficient_shifting}.
By using this reduced optimization, the optimal parameters are found as
\begin{equation}\label{eq:shifting_operator_predetermine}
\begin{split}
    \tau_q^{(1)}=&2h_{qq},\\
    \tau_{qrs}^{(2)}=&4g_{rsqq}-2g_{rqsq}.
\end{split}
\end{equation}
After this optimization, the entire number operators, along with a significant portion of one-body Hamiltonian, are discarded.
Consequently, only a part of the two-body Hamiltonian needs to be measured.

The shifting technique is also applicable to other algorithms that require the measurement of $a_{\hat{O}}(\tau):=\braket{\phi_0|\hat{O}e^{-i\hat{H}\tau}|\phi_0}$ for some Hermitian operator $\hat{O}$.
With the extended swap test, one can efficiently estimate $a_{\hat{O}}(\tau)$ from the measurement of $a_{\hat{O}-\hat{T}}(\tau)$ and $a_{\hat{\mathbb{I}}}(\tau)$, which are always measurable simultaneously.

\subsection{\label{subsec:Iterative Coefficient Splitting}Iterative Coefficient Splitting}
In contrast to the shifting technique, which minimizes the state-averaged costs (Eqs.\eqref{eq:asymptotic_LCU_variance} and \eqref{eq:asymptotic_FH_variance}), this section focuses on optimizing the state-dependent costs (Eqs.\eqref{eq:LCU_variance} and \eqref{eq:FH_variance}).
As one approach, we apply iterative coefficient splitting (ICS), initially designed for standard measurements \cite{Yen2023}, to the problem of measuring the matrix elements.

Given a qubit Hamiltonian, $\hat{H}=\sum_p \alpha_p \hat{P}_p$, where $\alpha_p\in\mathbb{R}$ and $\hat{P}_p\in\{\hat{\mathbb{I}}, \hat{\sigma}_x, \hat{\sigma}_y, \hat{\sigma}_z\}^{\otimes N_q}$, ICS seeks a decomposition into measurable operators that minimizes the total variance.
Note that measurable operators here involve not only Hermitian operators but also scaled unitaries, which are used for the Hadamard test.
As reviewed in Appendix \ref{sec:appendix_pauli_grouping}, such a decomposition in Pauli basis is written as:
\begin{equation}
    \hat{H}=\sum_{j=1}^{N_\zeta}\hat{N}_j = \sum_{j=1}^{N_\zeta}\sum_{p\in \mathcal{G}_j}\alpha^{(\mathcal{G}_j)}_p \hat{P}_p,
\end{equation}
where $\hat{N}_j$ are measurable operators and the corresponding sets of Pauli indices, $\mathcal{G}_j$, are predetermined by a decomposition algorithm like $\texttt{SORTED INSERTION}$.
Importantly, these $\mathcal{G}_j$ sets may not be disjoint, meaning the same operator $\hat{P}_p$ can belong to multiple sets.
Correspondingly, the coefficient $\alpha_p$ is split across these sets, satisfying the condition:
\begin{equation}\label{eq:ICS_CONST}
    \sum_{j:p\in \mathcal{G}_{j}}\alpha_p^{(\mathcal{G}_j)} = \alpha_p \quad \forall p.
\end{equation}

ICS leverages the flexibility of coefficient splitting to minimize the total variance, which is thus treated as a function of the split coefficients $\bm{\alpha}$ and the shot allocation $\bm{m}$:
\begin{equation}\label{eq:ICS_VAR}
\begin{split}
    \mathrm{Var}_{(\mathrm{ICS})}&(\bm{\alpha}, \bm{m})=\\
    &\frac{1}{M}\sum_{\substack{j=1,\cdots,N_\zeta \\ X\in \{R, I\}}}\frac{\mathrm{Var}[X_{0k;j};\bm{\alpha}^{(\mathcal{G}_j)}]}{m_{jX}},
\end{split}
\end{equation}
where the vectors of split coefficients are defined as $\bm{\alpha}:=(\bm{\alpha}^{(\mathcal{G}_1)}, \cdots, \bm{\alpha}^{(\mathcal{G}_{N_{\zeta}})})$ and $\bm{\alpha}^{(\mathcal{G}_j)}:=\{\alpha_{p}^{(\mathcal{G}_j)}:p\in\mathcal{G}_j\}$.
The partial variance for $\hat{N}_j$ (Eq.\eqref{eq:normal_partial_variance}) is expressed as a quadratic form in terms of the split coefficients $\bm{\alpha}^{(\mathcal{G}_j)}$:
\begin{equation}\label{eq:partial_variance_with_pauli_cov}
    \mathrm{Var}[X_{0k;j};\bm{\alpha}^{(\mathcal{G}_j)}]=\sum_{p, q\in \mathcal{G}_j}\alpha_{p}^{(\mathcal{G}_j)}\alpha_{q}^{(\mathcal{G}_j)}\mathrm{Cov}^{(X)}(\hat{P}_p, \hat{P}_q)_{\Phi_{0k}}.
\end{equation}
Here, the Pauli covariance for the real part is determined as:
\begin{equation}\label{eq:pauli_covariance_real}
\begin{split}
    \mathrm{Cov}^{(R)}(\hat{P}_p, \hat{P}_q)_{\Phi_{0k}} &= \braket{\Phi_{0k}|\hat{\mathbb{I}}\otimes\frac{1}{2}\{\hat{P}_p, \hat{P}_q\}|\Phi_{0k}}\\
    &-\mathrm{Re}[\braket{\phi_0|\hat{P}_p|\phi_k}]\mathrm{Re}[\braket{\phi_0|\hat{P}_q|\phi_k}].
\end{split}
\end{equation}
The detailed derivation is provided in Appendix \ref{sec:appendix_pauli_covariance}.
The covariance for the imaginary part is obtained by replacing $\mathrm{Re}$ with $\mathrm{Im}$ in Eq.\eqref{eq:pauli_covariance_real}.

To proceed with the optimization of Eq.\eqref{eq:ICS_VAR}, the covariances need to be estimated beforehand.
A direct and precise calculation of Eq.\eqref{eq:pauli_covariance_real} requires a state, $\ket{\phi_k}=e^{-i\hat{H}t_k}\ket{\phi_0}$, which is classically difficult to obtain.
Therefore, the covariance is approximated using configuration interaction single and double (CISD) ground state, $\ket{\mathrm{CISD}}$, and its energy, $E_{\mathrm{CISD}}$: 
\begin{equation}\label{eq:Phased_CISD_proxy}
\ket{\phi_k}\approx e^{-iE_{\mathrm{CISD}}t_k}\ket{\mathrm{CISD}}.
\end{equation}

Then, the optimization problem is given as:
\begin{equation}\label{eq:ICS_OPT}
    (\bm{\alpha}^{\star}, \bm{m}^{\star}) = \argmin_{\substack{\bm{\alpha}:\alpha_p=\sum_{j:p\in \mathcal{G}_j}\alpha_{p}^{(\mathcal{G}_j)} \,\,\forall p\\ \bm{m}:\sum_{jX}m_{jX}=1}}\mathrm{Var}_{(\mathrm{ICS})}(\bm{\alpha}, \bm{m}).
\end{equation}
Here, $\bm{\alpha}$ and $\bm{m}$, denote the split coefficients and the shot allocation, respectively.
Although optimizing both $\bm{\alpha}$ and $\bm{m}$ does not have a closed-form solution, the each step of alternating optimization\textemdash by fixing one variable ($\bm{\alpha}$ or $\bm{m}$) while optimizing the other\textemdash is a convex problem.
When $\bm{m}$ is held constant, constrained quadratic programming can be employed for optimizing $\bm{\alpha}$, because the variance is expressed as a quadratic form of $\bm{\alpha}$ as shown in Eq.\eqref{eq:partial_variance_with_pauli_cov}.
Conversely, when optimizing $\bm{m}$ with a fixed $\bm{\alpha}$, the Lagrangian multiplier method is utilized, which results in 
\begin{equation}\label{eq:opt_shot_alloc_ics}
    m_{jX}^{(\mathrm{opt})}\propto \mathrm{Var}[X_{0k;j};\bm{\alpha}^{(\mathcal{G}_j)}]^{1/2}.
\end{equation}

Overall, adapting the ICS method to the measurement problem for the matrix elements involves three additional key features compared to Ref.\cite{Yen2023}:
1) including scaled unitaries as measurable objects,
2) defining covariances between anticommuting Pauli operators,
3) employing the CISD proxy state for the time-evolved state.

\section{\label{sec:numerical_analysis}Numerical Results}

Here, we present numerical illustrations of our theoretical developments by examining the electronic structures of small molecules: $\ce{H2}$, $\ce{H4}$, $\ce{LiH}$, $\ce{BeH2}$ and $\ce{H2O}$, using the STO-3G basis set. The fermionic Hamiltonians are transformed to qubit operators by the Bravyi-Kitaev mapping with two-qubit reduction \cite{bravyi_kitaev, qubit_tapering}.

For the QKSD setting, we use the Hartree-Fock ground as the reference state $\ket{\phi_0}$.
The time step for the propagator is chosen as $\Delta_t=\pi/\Delta E_1$, following the choice in \cite[Theorem~3.1]{Theory_QKSD}, where $\Delta E_1$ represents the first spectral gap.
In practical scenario where the spectral gap is difficult to estimate in advance, a sufficiently large time step is often chosen to mitigate ill-conditioning, although this comes at the cost of increased circuit depth. The dependency of the conditioning on the time step has been numerically studied in \cite[Appendix A-2]{PRXQuantum.2.010333}.
Also, in order to focus on the finite sampling error, we assume that the exact propagator $\hat{B}=e^{-i\hat{H}\Delta_t}$ is available, which does not involve Trotterization error.
The QKSD order is set to $n=N_q+1$, where $N_q$ is the number of qubits.
In this setting, the overhead for the classical GEVP is exponentially small compare to the direct diagonalization, while the error induced by the projection onto the quantum Krylov space is exponentially small as shown in \cite[Theorem~3.1]{Theory_QKSD}.
We observed that the error caused by the Krylov projection is bounded as $|E_0^{(n)}-E_0|<10^{-4}\mathrm{~Ha}$ in the electronic structures of our interest, where $E_0$ is the true ground state energy of $\hat{H}$.

\begin{table}
    \centering
    \begin{tabular}{c|r|r|r|r|r}
    \multicolumn{1}{c|}{Norm (Hartree)} & \multicolumn{1}{c|}{$\ce{H2}$} & \multicolumn{1}{c|}{$\ce{H4}$} & \multicolumn{1}{c|}{$\ce{LiH}$} & \multicolumn{1}{c|}{$\ce{BeH2}$} & \multicolumn{1}{c}{$\ce{H2O}$} \\
    \Xhline{2\arrayrulewidth}
    $\|\bm{\beta}_{\texttt{SI}}(\hat{H})\|_1$          & 0.8405    & 6.0055    & 9.9902     & 16.4482     & 57.3794    \\
    $\|\bm{\beta}_{\texttt{SI}}(\hat{H}-\hat{T})\|_1$  & 0.1812    & 1.1278    & 0.4739     & 1.3582      & 2.0035     \\
    \hline
    Reduction(LCU, \%)                                 & 78.4      & 81.2      & 95.3       & 91.7        & 96.5     \\
    \Xhline{2\arrayrulewidth}
    $\|\bm{\gamma}_{\texttt{SI}}(\hat{H})\|_1$         & 0.7397    & 2.0310    & 2.5254     & 4.7003      & 21.9723    \\
    $\|\bm{\gamma}_{\texttt{SI}}(\hat{H}-\hat{T})\|_1$ & 0.1812    & 0.5288    & 0.3268     & 0.7857      & 1.1727     \\
    \hline
    Reduction(FH, \%)                                  & 75.50      & 74.0     & 87.1       & 83.3        & 94.7     \\
    \Xhline{2\arrayrulewidth}
    \end{tabular}
    \caption{
        LCU($\bm{\beta}$) and FH($\bm{\gamma}$) decomposition norms with and without shifting.
        \texttt{SORTED INSERTION} is adopted as decomposition algorithm.
        Shifting operators $\hat{T}$ are chosen as Eq.\eqref{eq:shifting_operator_electronic_structure} and optimized by the \texttt{POWELL} algorithm after assigning the parameters shown in Eq.\eqref{eq:shifting_operator_predetermine}.
    }
    \label{tab:norm_reduction}
\end{table}

\begin{table*}
    \centering
    \begin{tabular}{c|c|rr|rr|rr|rr|rr}
    \multicolumn{2}{c|}{$M\epsilon^2$ (Hartree$^2$)} & \multicolumn{2}{c|}{$\ce{H2}$} & \multicolumn{2}{c|}{$\ce{H4}$} & \multicolumn{2}{c|}{$\ce{LiH}$} & \multicolumn{2}{c|}{$\ce{BeH2}$} & \multicolumn{2}{c}{$\ce{H2O}$} \\
    \Xhline{2\arrayrulewidth}
    \multirow{5}{*}{LCU} & SI & 1.51 &(1.50) & 81.40 &(82.53) & 213.60 &(215.62) & 598.29 &(600.75) & 7265.04 &(7290.16) \\ \cline{2-12}
                        & ICS(True) & 0.88 &(0.87) & 69.84 &(69.77) & 185.04 &(182.27) & 534.46 &(534.83) & 6534.53 &(6508.27)\\ \cline{2-12}
                        & ICS(CISD) & 0.92 &(0.90) & 70.11 &(71.19) & 185.34 &(184.59) & 536.47 &(529.47) & 6550.79 &(6561.18)\\ \cline{2-12}
                        & Shift & 0.13 &(0.13) & 5.08 &(5.01) & 0.89 &(0.91) & 7.37 &(7.38) & 16.04 &(15.77) \\ \cline{2-12}
                        & Shift, ICS & 0.13 &(0.13) & 4.82 &(4.85) & 0.80 &(0.79) & 6.97 &(7.03) & 14.53 &(14.50) \\
    \Xhline{2\arrayrulewidth}
    \multirow{5}{*}{FH} & SI & 2.18 &(2.24) & 34.66 &(34.27) & 50.19 &(50.25) & 151.65 &(149.91) & 2284.67 &(2287.81) \\ \cline{2-12}
                        & ICS(True) & 1.29 &(1.29) & 18.32 &(18.32) & 26.43 &(26.29) & 88.00 &(87.51) & 1528.96 &(1529.87) \\ \cline{2-12}
                        & ICS(CISD) & 1.42 &(1.44) & 18.74 &(18.60) & 26.48 &(26.18) & 89.52 &(89.90) & 1535.81 &(1523.23) \\ \cline{2-12}
                        & Shift & 0.13 &(0.13) & 1.61 &(1.60) & 0.67 &(0.67) & 3.01 &(3.01) & 6.44 &(6.44) \\ \cline{2-12}
                        & Shift, ICS & 0.13 &(0.13) & 0.80 &(0.79) & 0.37 &(0.36) & 1.46 &(1.47) & 3.25 &(3.25) \\
    \Xhline{2\arrayrulewidth}
    \end{tabular}
    \caption{
        The measurement costs $M\epsilon^2$ for the (shifted) matrix elements $\braket{\phi_0|\hat{H}(-\hat{T})|\phi_k}$, averaged over $k$, are separately displayed for the cases with and without applying ICS and/or shift techniques from the measurement setting obtained from \texttt{SORTED INSERTION} (denoted as `SI').
        The costs are estimated by directly computing Eqs.\eqref{eq:FH_variance} and \eqref{eq:LCU_variance}.
        The values in the parenthesis are averaged empirical variances obtained by 1,000 independent runs of QKSD algorithm for each setting.
        The sub-optimal shot allocations are used if ICS is not employed, while shot allocations of the ICS output are adopted otherwise.
        The result of ICS based on the true state $\ket{\phi_k}$ and phased CISD proxy are shown.
        Note that ICS by true state is not practically achievable.
    }
    \label{tab:measurement_cost}
\end{table*}

The reduction in the norm achieved through the shifting operator in Eq.\eqref{eq:shifting_operator_electronic_structure} is presented in Table \ref{tab:norm_reduction}.
Overall, the shifting technique reduced the norm more than 74\%.
The relative reductions are larger in the LCU case because the shifting removes large $\hat{Z}$-type Pauli operators.
These operators cannot be grouped together in the LCU decomposition, leading to a larger norm when they were not removed.
The details are provided in Appendix \ref{sec:appendix_pauli_grouping}.
However, the resulting costs of LCU remain higher than those for FH, which implies that FH allows more efficient measurement.

The exact and empirical measurement costs for each scenario, both with and without the techniques described in Section \ref{sec:Sampling Cost Reduction}, are tabulated in Table \ref{tab:measurement_cost}.
The measurement costs obtained by the experiments approximate the exact costs described in Eqs. \eqref{eq:LCU_variance} and \eqref{eq:FH_variance} within the error caused by the finite number of experiments.
Generally, the cost reduction tends to increase with the system size.
Also, a significant portion of the reduction is attributed to the shifting technique, which correlates closely with the squared reduction ratio in Table \ref{tab:norm_reduction}.
This correlation suggests that $M\epsilon^2 \propto \|\bm{\zeta}\|_1^2$, aligning with the relationship previously established in Eq.\eqref{eq:zeta_cost}.

Furthermore, we validated the approximation of Pauli covariance (Eq.\eqref{eq:pauli_covariance_real}) using the CISD proxy (Eq.\eqref{eq:Phased_CISD_proxy}) by comparing ICS results obtained with both the proxy and the true Krylov basis state, as shown in Table \ref{tab:measurement_cost}.

\begin{figure}
    \centering
    \includegraphics[width=0.95\linewidth]{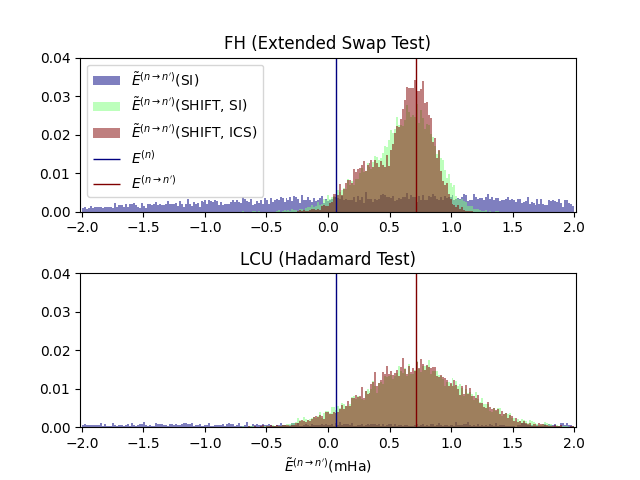}
    \caption{Histogram of the error of the estimated ground state energies ($\tilde{E}^{(n\rightarrow n')}_0-E_{\mathrm{FCI}}$) of the electronic structure Hamiltonian of $\ce{H2O}$ with different measurement settings obtained from QKSD algorithm with the thresholding by Eq. \eqref{eq:optimal_truncation} and the finite number of shots of $M=10^8$.
    The horizontal axis represents the errors in the atomic unit (mHa), and the vertical one denotes the frequency of the each histogram bin.
    The histogram is plotted using the perturbed QKSD energies from $10,000$ independent and identical experiments.
    Here, the optimal thresholding of Eq.\eqref{eq:optimal_truncation} is applied to mitigate the sampling error, further reducing the Krylov dimension from $n=9$ to $n'=3$.
    In the FH case, two different values of $n'=2,3$ were observed across the random experiments, resulting in the two peaks in the histogram.
    $\tilde{E}^{(n\rightarrow n')}$ and ${E}^{(n\rightarrow n')}$ denote the QKSD ground state energies with and without considering the effect of sampling error}, respectively.
    $E^{(n)}$ represents the QKSD ground state energy without error or thresholding.
    \label{fig:QKSD_histogram}
\end{figure}

However, we observed that for the LCU cases, the ICS method performed less effectively than the FH cases.
As shown in Appendix \ref{sec:appendix_pauli_grouping}, a scaled unitary fragment, $\beta_j\hat{U}_j$, is constructed by grouping mutually anticommuting Pauli operators, while commuting ones form a Hamiltonian fragment, $\hat{H}_j$.
In general, because the anticommutation between Pauli operators occurs less frequently, there are fewer opportunities for the Pauli operators to be grouped to form a unitary.
This makes the coefficient splitting with LCU less effective.

As shown in Fig.\ref{fig:QKSD_histogram}, we compare the perturbed ground state energies of the electronic Hamiltonian of $\ce{H2O}$, obtained by the QKSD algorithm, with and without applying the reduction techniques.
We observed the unperturbed QKSD energy, $E^{(n)}$, is close to the full configuration interaction (FCI) energy ($|E^{(n)} - E_{\mathrm{FCI}}|\approx 0.1 \mathrm{mHa}$).

We also employed a classical postprocessing called basis thresholding to alleviate the numerical instability of GEVP \cite{Lee2024samplingerror, Theory_QKSD}.
Since small singular values in $\bm{S}$ significantly amplify the perturbation to the eigenvalue $E^{(n)}$, the thresholding technique further projects the GEVP onto the singular basis of $\bm{S}$ with corresponding singular values larger than a certain value $\epsilon_{\mathrm{th}}>0$.
However, the thresholding also discards the information about the eigenstates, not only the error, which biases thresholded QKSD energy.
Thus, by adjusting $\epsilon_{\mathrm{th}}$, the thresholding establishes a trade-off, reducing the effect of the statistical error in $\bm{H}$ and $\bm{S}$ while introducing additional projection error.
Within this trade-off, the optimal $\epsilon_{\mathrm{th}}$ was heuristically found in Ref.\cite{Lee2024samplingerror} to be: 
\begin{equation}\label{eq:optimal_truncation}
\epsilon_{\mathrm{th}}=\tilde{O}(n/\sqrt{M_{\bm{S}}}),    
\end{equation}
where $M_{\bm{S}}$ is the number of shots used to construct $\bm{S}$ with the Hadamard or extended swap test.
Note that for the case of errors other than the sampling error are present, an automated thresholding \cite{yoshioka2024diagonalization} can be adopted.
We denote $n'$ as the dimension of the thresholded problem and $E^{(n\rightarrow n')}$ as corresponding eigenvalue.
%

Despite the bias caused by the projection error, the effect of the matrix perturbation is minimized.
In the FH case, the application of the shifting technique and ICS resulted in the perturbed QKSD solution being concentrated within chemical accuracy ($|\tilde{E}_0^{(n\rightarrow n')}-E_{\mathrm{FCI}}| < 1.6\mathrm{mHa}$), whereas most results without these techniques deviated beyond chemical accuracy.


\section{\label{sec:conclusion}Conclusion}
In this work, we analyzed the finite sampling errors that arise when projecting the Hamiltonian onto the quantum Krylov subspace with quantum algorithms.
The measurement cost analyses of two scenarios, LCU and FH decomposition, converge to a unified perspective, where the decomposition rules and circuit construction are different.
We also showed that the expected cost is analogous to the LCU 1-norm, which has the same definition in the context of block encoding \cite{LCU_childs}.
Such analogies enable the translation of methods originally motivated to reduce the costs in LCU simulation \cite{Loaiza_BLISS, Loaiza_2023} and the expectation measurement \cite{Crawford2021efficientquantum, Yen2023} to the problem of measuring matrix elements.
Especially, adopting the symmetry shift \cite{Loaiza_BLISS, Loaiza_2023} to the measurement problem eases the constraint on the shift operator, and thus provides larger cost reduction.
Although the shifting technique is more effective in the LCU case, the measurement cost is observed to be lower in the FH case.

Despite of the effort to reduce the measurement cost, achieving the chemical accuracy by QKSD in practice remains still challenging, as the corresponding GEVP is often ill-conditioned.
Note that in our application of QKSD to the $\ce{H2O}$ system, the results were fitted within the chemical accuracy by using $10^8$ shots, which can be considered expensive.
In classical Krylov subspace diagonalization, the perturbation on the GEVP matrices mainly depends on the round off or floating point error, which decreases exponentially to the number of bits of the data type, although the calculation of the matrix elements takes exponentially long time.
On the other hand, governed by the standard quantum limit, the matrix perturbation in QKSD decays with the square root of the number of measurements, which is much slower than the classical KSD, while each measurement takes polynomial time.
Therefore, if exponential precision for the matrix element is required because of the ill-conditioning, it is not yet obvious that QKSD is superior to the classical counterpart in terms of running time to achieve a certain precision in the estimated eigenvalues.

As discussed in Section \ref{sec:Measurement of QKSD Matrix Elements}, we observed that the extended swap test, enabled by FH decomposition, allows for the simultaneous measurement of both the overlap and Hamiltonian matrices:
\begin{gather*}
    \bm{S}_{0k}=\braket{\Phi_{0k}|\hat{\mathbb{I}}\otimes(\hat{\sigma}_x+i\hat{\sigma}_y)|\Phi_{0k}},\\
    \bm{H}_{0k}=\braket{\Phi_{0k}|\hat{H}\otimes(\hat{\sigma}_x+i\hat{\sigma}_y)|\Phi_{0k}},
\end{gather*}
due to commutativity between $\hat{\mathbb{I}}\otimes\hat{\sigma}_{x(y)}$ and $\hat{H}\otimes\hat{\sigma}_{x(y)}$.
This idea can be extended to the other EFTQC algorithms.
Based on our knowledge, recently developed EFTQC algorithms other than QKSD focus on extracting the spectrum only from the autocorrelation function, $a(t):=\braket{\phi_0|e^{-i\hat{H}t}|\phi_0}$.
However, those algorithms can be more refined by adopting the simultaneous measurement of $a_{\hat{O}}(t):=\braket{\phi_0|\hat{O}e^{-i\hat{H}t}|\phi_0}$ for some observable $\hat{O}$.
For instance, the first derivative of $a(t)$ can be directly calculated from $a_{\hat{H}}(t)$.
This measurement, which can be done precisely if the techniques introduced in this work are adopted, only requires overheads of $O(N_q^2)$ Clifford operations.
These operations do not need additional logical qubits and present an endurable cost for EFTQCs.
Given this rationale, our future research direction involves exploring EFTQC algorithms including $a_{\hat{O}}(t)$ in the spectrum extraction.

\section*{Data Availability Statement}
    The Python package for the techniques introduced in this work (shifting and iterative coefficient splitting) is available at \cite{ofex}.
    Furthermore, the raw data for numerical results are available at \cite{reproduce} along with the reproducing scripts.

\section*{Conflicts of Interest}
    There are no conflicts of interest to declare.

\section*{Author Contributions}

    GL: Conceptualization, Methodology, Software, Validation, Formal Analysis, Investigation, Data Curation, Writing – Original Draft Preparation,  Writing – Review \& Editing, Visualization, Funding Acquisition.
    SC: Methodology, Software.
    JH: Conceptualization, Investigation, Resources, Writing – Review \& Editing, Supervision, Project Administration, Funding Acquisition.
    AFI: Conceptualization, Methodology, Formal Analysis, Investigation, Resources, Writing – Review \& Editing, Visualization, Supervision, Project Administration, Funding Acquisition.

\begin{acknowledgments}
A.F.I. acknowledges financial support from
the Natural Sciences and Engineering Council of Canada
(NSERC).

G.L. acknowledges the support from the Education and Training Program of the Quantum Information Research Support Center, funded through the National Research Foundation of Korea (NRF) by the Ministry of Science and ICT (MSIT) of Korean government (NRF-2021M3H3A1036573).
Additionally, this work was supported by Basic Science Research Program through the National Research Foundation of Korea (NRF), funded by the Ministry of Education, Science and Technology (NRF-2022M3H3A106307411, NRF-2023M3K5A1094813).
This work was also partly supported by Institute for Information \& communications Technology Promotion (IITP) grant funded by the Korea government(MSIP) (No. 2019-0-00003, Research and Development of Core technologies for Programming, Running, Implementing and Validating of Fault-Tolerant Quantum Computing System). 
The Ministry of Trade, Industry, and Energy (MOTIE), Korea, also partly supported this research under the Industrial Innovation Infrastructure Development Project (No. RS-2024-00466693).
\end{acknowledgments}
    
\appendix

\section{Averaged Variance on $\bm{H}_{0k}$}\label{sec:appendix_asymptotic_proof}
In this section, we derive the state-independent variance, represented by Eqs.\eqref{eq:asymptotic_LCU_variance} and \eqref{eq:asymptotic_FH_variance}, by averaging the states $\ket{\phi_0}$ and $\ket{\phi_k}$ over the independent uniform Haar distributions $(\phi_0\sim \mathcal{H}(d), \phi_k \sim \mathcal{H}(d))$.
Before proceeding with the derivation, we introduce three lemmas that will be instrumental in this process.
Throughout this section, $\mathbb{E}_{\phi\sim \mathcal{H}(d)}[\cdot]$ is abbreviated as $\mathbb{E}_{\phi}[\cdot]$ unless otherwise mentioned.
\begin{lemma}\label{lemm:haar_square_independent}
    For any normal operator $\hat{A} \in \mathbb{C}^{d\times d}$, the following equality holds:
    \begin{equation}
        \mathbb{E}_{\phi_1, \phi_2}\left[ \left|\bra{\phi_1} \hat{A} \ket{\phi_2}\right|^2 \right] = \frac{1}{d^2}\mathrm{Tr}[\hat{A}^{\dagger}\hat{A}].
    \end{equation}
\end{lemma}
\begin{proof}
    For any operator $\hat{X}\in \mathbb{C}^{d\times d}$, the averaged conjugation is known as
    \begin{align}
        \mathbb{E}_{\hat{U}\sim \mathcal{U}(d)}\left[\hat{U}^{\dagger} \hat{X} \hat{U}\right] = \frac{\hat{I}}{d} \mathrm{Tr}[\hat{X}],
    \end{align}
    as a consequence of Schur's lemma and the left and right invariance of the Haar measure, where $\mathcal{U}(d)$ is the uniform Haar distribution over the unitary group of dimension $d$.
    Therefore, we can state
    \begin{align}\label{eq:average_state_expectation}
        \mathbb{E}_{\phi}\left[\bra{\phi} \hat{X} \ket{\phi}\right] = \frac{\mathrm{Tr}[\hat{X}]}{d}.
    \end{align}
    Finally, by applying Eq.\eqref{eq:average_state_expectation} consecutively, we have:
    \begin{align*}
    \mathbb{E}_{\phi_1, \phi_2}\left[ \left|\bra{\phi_1} \hat{A} \ket{\phi_2}\right|^2 \right]
    =& \mathbb{E}_{\phi_1, \phi_2}\left[ \bra{\phi_1} \hat{A} \ket{\phi_2}\bra{\phi_2} \hat{A}^{\dagger} \ket{\phi_1} \right]\\
    =& \frac{1}{d}\mathbb{E}_{\phi_2}\left[\mathrm{Tr}[\hat{A}\ket{\phi_2}\bra{\phi_2}\hat{A}^\dagger]\right] \\
    =& \frac{1}{d}\mathbb{E}_{\phi_2}\left[\bra{\phi_2}\hat{A}^\dagger\hat{A}\ket{\phi_2}\right] \\
    =& \frac{1}{d^2} \mathrm{Tr}[\hat{A}^\dagger \hat{A}].
    \end{align*}
\end{proof}

\begin{lemma}\label{lemm:haar_reim_square_independent}
    For any normal operator $\hat{A}\in \mathbb{C}^{d\times d}$, the following equality holds:
    \begin{equation}\label{eq:haar_reim}
        \mathbb{E}_{\phi_1, \phi_2}\left[\mathrm{Re}[\braket{\phi_1|\hat{A}|\phi_2}]^2\right] = \mathbb{E}_{\phi_1, \phi_2}\left[\mathrm{Im}[\braket{\phi_1|\hat{A}|\phi_2}]^2\right].
    \end{equation}
\end{lemma}
\begin{proof}
    Because of the unitary-invariant property of Haar measure, $\mathbb{E}_{\phi}[f(\ket{\phi})]$ is always identical to $\mathbb{E}_{\phi}[f(e^{i\pi/2}\ket{\phi})]=\mathbb{E}_{\phi}[f(i\ket{\phi})]$ for any function $f$.
    Therefore, the following proves Eq.\eqref{eq:haar_reim}:
    \begin{align}
        \mathbb{E}_{\phi_1, \phi_2}&\left[\mathrm{Re}[\braket{\phi_1|\hat{A}|\phi_2}]^2 - \frac{1}{2}\left|\braket{\phi_1|\hat{A}|\phi_2}\right|^2\right]\\
        =&\frac{1}{4}\mathbb{E}_{\phi_1,\phi_2}\left[ \braket{\phi_1|\hat{A}|\phi_2}^2+\braket{\phi_2|\hat{A}^{\dagger}|\phi_1}^2 \right]\label{eq:lemm_2_pf1}\\
        =&\frac{1}{4}\mathbb{E}_{\phi_1,\phi_2}\left[ \braket{\phi_1|i\hat{A}|\phi_2}^2+\braket{\phi_2|(-i)\hat{A}^{\dagger}|\phi_1}^2 \right]\label{eq:lemm_2_pf2}\\
        =&\mathbb{E}_{\phi_1, \phi_2}\left[\mathrm{Im}[\braket{\phi_1|\hat{A}|\phi_2}]^2 - \frac{1}{2}\left|\braket{\phi_1|\hat{A}|\phi_2}\right|^2\right],\label{eq:lemm_2_pf3}
    \end{align}
    where Eqs.\eqref{eq:lemm_2_pf1} and \eqref{eq:lemm_2_pf3} are obtained using $\mathrm{Re}[z]^2=\frac{1}{4}[z^2+z^{*2}+2|z|^2]$ and $\mathrm{Im}[z]^2=\frac{1}{4}[-z^2-z^{*2}+2|z|^2]$, respectively, for $z\in\mathbb{C}$.
    Additionally, Eq.\eqref{eq:lemm_2_pf2} is derived by replacing $\ket{\phi_2}$ with $i\ket{\phi_2}$.
\end{proof}

\begin{lemma}\label{lemm:haar_square_diag}
    For any normal operator $\hat{A} \in \mathbb{C}^{d\times d}$, the following equality holds:
    \begin{equation}
        \mathbb{E}_{\phi}\left[ \left|\bra{\phi} \hat{A} \ket{\phi}\right|^2 \right] = \frac{\mathrm{Tr}[\hat{A}^{\dagger}\hat{A}]+|\mathrm{Tr}[\hat{A}]|^2}{d(d+1)}.
    \end{equation}
\end{lemma}
\begin{proof}
    Because $\hat{A}$ is normal, we can perform the eigendecomposition to an arbitrary state as $\ket{\phi} = \sum_{i=1}^{d} \gamma_{i} \ket{\psi_{i}}$, where $\hat{A}\ket{\psi_i}=a_{i}\ket{\psi_{i}}$ and $\gamma_i=\braket{\psi_i|\phi}$.
    Then the Haar averaging over $\ket{\phi}$ is identical to the averaging $\{\gamma_{i}\}_{i=1}^{d}$ over $d$-dimensional complex unit sphere, $\mathbb{E}_{\phi}[\cdot] = \mathbb{E}_{\{\gamma_i\}\sim \mathbb{CP}^{d-1}}[\cdot]$, where $\mathbb{CP}$ denotes the complex projective space.
    Therefore, it can be shown that
    \begin{align}
        \mathbb{E}_{\phi}&\left[ \left| \bra{\phi}\hat{A}\ket{\phi} \right|^2 \right]\\
        =& \mathbb{E}_{\{\gamma_i\} \sim \mathbb{CP}^{d-1}}\left[\left|\sum_{i=1}^d a_i |\gamma_i|^2 \right|^2 \right]\\
        =& \sum_i |a_{i}|^2\mathbb{E}_{\bm{\gamma}}[|\gamma_i|^4]+\sum_{i\neq j}a_i^*a_j\mathbb{E}_{\bm{\gamma}}[|\gamma_i|^2|\gamma_j|^2]\\
        =& \frac{1}{d(d+1)}\left(2\sum_i |a_i|^2 + \sum_{i\neq j}a_{i}^* a_{j}\right) \label{eq:hom_avg_csphere}\\
        =& \frac{1}{d(d+1)}\left(\sum_i |a_i|^2 + \left(\sum_{i}a_{i}^*\right) \left(\sum_{i}a_{i}\right)\right)\\
        =& \frac{1}{d(d+1)}\left(\mathrm{Tr}[\hat{A}^{\dagger}\hat{A}] + \mathrm{Tr}[\hat{A}^\dagger]\mathrm{Tr}[\hat{A}]\right),
    \end{align}
    where $\mathbb{E}_{\bm{\gamma}}[\cdot]$ denotes $\mathbb{E}_{\{\gamma_i\}\sim \mathbb{CP}^{d-1}}[\cdot]$ and Eq.\eqref{eq:hom_avg_csphere} holds because of \cite[Theorem 2.6]{https://doi.org/10.1002/jcd.21379}.
\end{proof}

For the case of LCU decomposition, the Haar-averaged partial variance in Eq.\eqref{eq:LCU_partial_variance} is identically determined over $X$ and $j$ as
\begin{equation}
    \mathbb{E}_{\phi_0,\phi_k}\left[\mathrm{Var}[\hat{O}_{X}]_{\Phi_{0k;j}}\right]=1-\frac{1}{2d},
\end{equation}
by applying Lemmas \ref{lemm:haar_square_independent} and \ref{lemm:haar_reim_square_independent} with $\hat{A}=\hat{U}_j$.
Therefore, the sub-optimal shot allocation, 
\begin{equation}\label{eq:app_LCU_subopt_alloc}
    m_{jX}^{(\mathrm{subopt})}=\frac{\beta_j}{2\|\bm{\beta}\|_1},
\end{equation}
is obtained by replacing $\mathrm{Var}[\hat{O}_X]_{\Phi_{0k;j}}$ in Eq.\eqref{eq:LCU_optimal_allocation} with the expected variance, where the denominator is set to satisfy the normalization constraint, $\sum_{jX}m_{jX}^{(\mathrm{subopt})}=1$.

The variance with the sub-optimal shot allocation is derived from Eq.\eqref{eq:LCU_variance} by assigning Eq.\eqref{eq:app_LCU_subopt_alloc}, which is
\begin{equation}\label{eq:asymptotic_var_pf_lcu0}
\begin{split}
        \mathrm{Var}^{(\mathrm{LCU})}&[\bm{H}_{0k};\bm{m}^{(\mathrm{subopt})}]\\&= \frac{1}{M}\sum_{j=1}^{N_\beta}2\|\bm{\beta}\|_1^2\left(2 - |\braket{\phi_0|\hat{U}_j|\phi_k}|^2\right).
\end{split}
\end{equation}
Finally, by Lemma \ref{lemm:haar_square_independent}, the Haar averaging of $|\braket{\phi_0|\hat{U}_j|\phi_k}|^2$ both over $\phi_0$ and $\phi_k$ results in Eq.\eqref{eq:asymptotic_LCU_variance}.

Furthermore, because $\ket{\phi_k}=\hat{B}^k\ket{\phi_0}$ is an evolved state from $\ket{\phi_0}$, we can consider averaging the variance over single state $\ket{\phi_0}$, fixing the evolution operator, $\hat{B}^k$.
Then from Eq.\eqref{eq:asymptotic_var_pf_lcu0}, $|\braket{\phi_0|\hat{U}_j|\phi_k}|^2=|\braket{\phi_0|\hat{U}_j\hat{B}^k|\phi_0}|^2$ needs to be averaged over $\ket{\phi_0}$.
Using Lemma \ref{lemm:haar_square_diag} with $\hat{A}=\hat{U}_j \hat{B}^k$, we can show that 
\begin{equation*}
\mathbb{E}_{\phi_0}[|\braket{\phi_0|\hat{U}_j\hat{B}^k|\phi_0}|] = \frac{d+|\mathrm{Tr}[\hat{U}_j\hat{B}^k]|^2}{d(d+1)},
\end{equation*}
and thus
\begin{equation*}
\frac{1}{d+1}\le\mathbb{E}_{\phi_0}[|\braket{\phi_0|\hat{U}_j\hat{B}^k|\phi_0}|]\le 1,
\end{equation*}
because $0 \le |\mathrm{Tr}[\hat{U}_j\hat{B}^k]|^2\le d^2$.
Finally, the total averaged variance is bounded by
\begin{equation}
\begin{split}
    \frac{2\|\bm{\beta}\|_1^2}{M}\le\mathbb{E}_{\phi_0}\left[\mathrm{Var}^{(\mathrm{LCU})}\right]\le\frac{2\|\bm{\beta}\|_1^2}{M}\left(2 - \frac{1}{d+1}\right),
\end{split}
\end{equation}
which is approximately less than the expected variance considering the two states independently (Eq.\eqref{eq:asymptotic_LCU_variance}.)

In FH case, the shots for real and imaginary parts are identically allocated ($m_{jR}=m_{jI}=:m_j/2$) because of Lemma \ref{lemm:haar_reim_square_independent}, similar to the LCU case.
Therefore, the variances for the real and imaginary parts in Eq.\eqref{eq:FH_partial_variance} can be directly added, yielding
\begin{equation}
\begin{split}
    \mathrm{Var}^{(\mathrm{FH})} =& \sum_{j=1}^{N_{\gamma}}\frac{2}{m_j}\left(\mathrm{Var}[\hat{O}_{jR}]_{\Phi_{0,k}}+ \mathrm{Var}[\hat{O}_{jI}]_{\Phi_{0,k}}\right)\\
    =&\sum_{j=1}^{N_{\gamma}}\frac{2}{m_j}\left(2\braket{\hat{O}^2_{jX}}_{\Phi_{0k}} - |\braket{\phi_0|\hat{H}_j|\phi_k}|^2\right).
\end{split}
\end{equation}
In Eq.\eqref{eq:FH_partial_variance}, the expected second moment is determined as $\mathbb{E}_{\phi_0, \phi_k}[\braket{\hat{O}_{jX}^2}_{\Phi_{0k}}]=\mathrm{Tr}[\hat{H}_j^2]/d$, while the last term, $\mathbb{E}_{\phi_0, \phi_k}[|\braket{\phi_0|\hat{H}|\phi_k}|^2]=\mathrm{Tr}[\hat{H}_j^2]/d^2$ is obtained using Lemma \ref{lemm:haar_square_independent}.
Finally, the expected variance of $j^{\textit{th}}$ fragment is derived as below:
\begin{equation}
\begin{split}
    \mathbb{E}_{\phi_0,\phi_k}\left[ \mathrm{Var}[\hat{O}_{jR}]_{\Phi_{0,k}}+ \mathrm{Var}[\hat{O}_{jI}]_{\Phi_{0,k}}\right]\\
    = \frac{1}{d}\mathrm{Tr}[\hat{H}_j^2]\left(2-\frac{1}{d}\right),
\end{split}
\end{equation}
which results in a total averaged variance of Eq.\eqref{eq:asymptotic_FH_variance} with the sub-optimal allocation $m_j \propto \mathrm{Tr}[\hat{H}_j^2]^{1/2}$.

Similar to the LCU case, if we take the expectation only on $\ket{\phi_0}$, we have
\begin{equation}
    \mathbb{E}_{\phi_0}\left[\mathrm{Var}^{(\mathrm{FH})}\right] \le \frac{2\|\bm{\gamma}\|^2}{M}\left(2 - \frac{1}{d+1}\right).
\end{equation}

\section{Grouping Pauli Operators}\label{sec:appendix_pauli_grouping}
In this section, we review the partitioning of qubit Hamiltonian in a form of LCU and FH.
Also, the expressions of $\|\bm{\beta}\|_1$ and $\|\bm{\gamma}\|_1$ are presented with respect to the Pauli coefficients.

A qubit Hamiltonian $\hat{H}$ is expressed as
\begin{equation}\label{eq:pauli_hamiltonian}
    \hat{H} = \sum_{p=1}^{N_P}\alpha_p \hat{P}_p,
\end{equation}
where $\hat{P}_p\in \{\hat{I}, \hat{\sigma}_x, \hat{\sigma}_y, \hat{\sigma}_z\}^{\otimes N_q}$ is an $N_q$-qubit Pauli operator and its coefficient is $\alpha_p \in \mathbb{R}$.

First, we review the derivation of LCU in Pauli basis(\cite{LCU_Izmaylov, LCU_Love}), which is described as
\begin{equation}\label{eq:pauli_LCU}
    \hat{H}=\sum_{j=1}^{N_\beta} \beta_j \hat{U}_j = \sum_{j=1}^{N_\beta}\sum_{p\in \mathcal{A}_j}\alpha_{p}^{(\mathcal{A}_j)}\hat{P}_p.
\end{equation}
Here, a Pauli term $\alpha_p \hat{P}_p$ can be separated to the multiple groups $\{\mathcal{A}_j\}$.
Thus, $\alpha_p=\sum_{j:p\in\mathcal{A}_j}\alpha_p^{(\mathcal{A}_j)}$ for all $1\le p \le N_p$, should be imposed for the partitioned coefficients.
Also, if the anticommutation conditions of different Pauli operators within a group holds:
\begin{equation}
    \{\hat{P}_p, \hat{P}_q\}=2\delta_{pq}\hat{I}\quad \forall p, q \in \mathcal{A}_j, \forall j,
\end{equation}
each partition, $\sum_{p\in\mathcal{A}_j}\alpha_p^{(\mathcal{A}_j)}\hat{P}_p$, becomes a scaled unitary, $\beta_j \hat{U}_j$.
Then, $\beta_j^2=\sum_{p\in \mathcal{A}_j}\alpha_p^{(\mathcal{A}_j) 2}$ is found consequently, and the norm becomes
\begin{equation}\label{eq:Pauli_norm_LCU}
    \|\bm{\beta}\|_1 = \sum_{j=1}^{N_\beta}\left(\sum_{p\in \mathcal{A}_j}\alpha_p^{(\mathcal{A}_j)2}\right)^{1/2}.
\end{equation}
Furthermore, a circuit realization of controlled $\hat{U}_j$ is known, which requires $O(N_q|\mathcal{A}_j|)$ two-qubit gates \cite{LCU_Love}.

Now, for the case of FH, Pauli operators are grouped together if all pairs of the operators in a group commute:
\begin{equation}\label{eq:pauli_FH}
    \hat{H}=\sum_{j=1}^{N_\gamma} \hat{H}_j = \sum_{j=1}^{N_\gamma}\sum_{p\in \mathcal{C}_j}\alpha_{p}^{(\mathcal{C}_j)}\hat{P}_p,
\end{equation}
where
\begin{equation}
    \left[\hat{P}_p, \hat{P}_q\right] = 0 \quad \forall p, q \in \mathcal{C}_j, \forall j,
\end{equation}
and $\alpha_p=\sum_{j:p\in\mathcal{C}_j}\alpha_p^{(\mathcal{C}_j)}$.
For each of $\hat{H}_j$, one can construct a diagonalizing Clifford circuit, $\hat{V}_j$, taking $O(N_q\min (N_q, |\mathcal{C}_j|))$ two-qubit gates.
Furthermore, one can find the FH decomposition norm, $\|\bm{\gamma}\|_1$ in terms of Pauli coefficients:
\begin{equation}\label{eq:Pauli_norm_FH}
    \|\bm{\gamma}\|_1=\frac{1}{\sqrt{d}}\sum_{j=1}^{N_\gamma}\mathrm{Tr}[\hat{H}_j^2]^{1/2}=\sum_{j=1}^{N_\gamma}\left(\sum_{p\in\mathcal{C}_j}\alpha_p^{(\mathcal{C}_j)2}\right)^{1/2}.
\end{equation}
Here, note that we used $\frac{1}{d}\mathrm{Tr}[\hat{P}_p \hat{P}_q]=\delta_{p,q}$ because the trace of Pauli operators except the identity is zero.

The two norms described in Eqs.\eqref{eq:Pauli_norm_LCU} and \eqref{eq:Pauli_norm_FH} shares a common expression as the sum of the L2 norms of Pauli coefficients.
We can consider decomposing Eq.\eqref{eq:pauli_hamiltonian} without grouping, $\hat{H}_j=\beta_j\hat{U}_j=\alpha_j \hat{P}_j$, which is also an LCU and an FH simultaneously.
However, in such case, the norm $\|\bm{\beta}\|_1=\|\bm{\gamma}\|_1=\|\bm{\alpha}\|_1$ becomes larger than that of grouped Pauli case because the sum of L2 norms of grouped Pauli coefficients is smaller than the L1 norm.

Pauli groupings are not unique and can be translated to a clique covering problem on the (anti-)commutation graph.
The commutation graph, denoted as $G_{\mathcal{C}}=(V, E_{\mathcal{C}})$, encodes the commutativity between Pauli operators.
Specifically, the nodes in $V$ correspond to individual Pauli operators: $V=\{v_p:1\le p \le N_P\}$ with a node weight function $w(v_p)=\alpha_p$.
Also, undirected and unweighted edges connect nodes whose corresponding operators commute ($E_{\mathcal{C}}=\{(v_p, v_q):[\hat{P}_p, \hat{P}_q]=0\}$).
The anticommutation graph, $G_{\mathcal{A}}=(V, E_{\mathcal{A}})$, shares the same node set with $G_{\mathcal{C}}$ but connects nodes with anticommuting operators ($E_{\mathcal{A}}=\{(v_p, v_q):\{\hat{P}_p, \hat{P}_q\}=0\}$).
Because any two Pauli operators either commute or anticommute, these graphs are complements, meaning that $E_{\mathcal{A}}^{\mathrm{c}}=E_{\mathcal{C}}$.

In such setting, minimizing Eqs.\eqref{eq:Pauli_norm_LCU} and \eqref{eq:Pauli_norm_FH} translates to finding a clique covering that minimizes the sum of the clique weights.
Each weight is defined as a L2 norm of node weights covered by each clique.
Like other clique covering problems, this is an NP-hard problem.
However, a heuristic and greedy algorithm called $\texttt{SORTED INSERTION}$ often outperform other heuristic algorithms as demonstrated in \cite{Crawford2021efficientquantum}.
The original work on $\texttt{SORTED INSERTION}$ aimed to minimize the cost of measuring standard expectation values, so it only considered the commutation graph.
However, we extend the algorithm to the anticommutation case to reduce the cost of the LCU measurements by simply modifying its grouping condition.

In the numerical results of Section \ref{sec:numerical_analysis}, the $\hat{Z}$-type operators, which dominate the example Hamiltonians, correspond to the number operators in the fermionic representation.
In the LCU decomposition, these operators are not grouped due to their mutual commutativity, resulting in a larger LCU norm than FH norm ($\|\bm{\beta}_{\mathtt{SI}}(\hat{H})\|\ge\|\bm{\gamma}_{\mathtt{SI}}(\hat{H})\|$) because large coefficients are treated separately.
However, the shifting technique effectively eliminates these operators, significantly reducing the norm $\|\bm{\beta}_{\mathtt{SI}}(\hat{H}-\hat{T})\|$.
In the FH decomposition, $\hat{Z}$-type operators are naturally grouped together, which contribute to the norm less significantly than the LCU case.
Consequently, the impact of the shifting technique is less pronounced in the FH case compared to the LCU case.
Nevertheless, the shifted FH decomposition still achieves a smaller norm than the shifted LCU decomposition.

\section{Pauli Covariance}\label{sec:appendix_pauli_covariance}
In this section derives the Pauli covariance for measuring the matrix element (Eq.\eqref{eq:pauli_covariance_real}).
Starting from the partial variance in Eq.\eqref{eq:normal_partial_variance}, and substituting $\hat{N}_j=\sum_{p\in\mathcal{G}_j}\alpha_p^{(\mathcal{G}_j)}\hat{P}_p$, which is analogous to Eq.\eqref{eq:pauli_LCU} or \eqref{eq:pauli_FH}, we obtain:
\begin{widetext}
\begin{equation}\label{eq:pauli_covariance_derive1}
    \mathrm{Var}[R_{0k;j}]=\sum_{p,q \in \mathcal{G}_j}
    \alpha_p^{(\mathcal{G}_j)} \alpha_q^{(\mathcal{G}_j)}
    \left[\frac{1}{2}\left(\braket{\phi_0|\hat{P}_p\hat{P}_q|\phi_0} + \braket{\phi_k|\hat{P}_p\hat{P}_q|\phi_k} \right) - \mathrm{Re}[\braket{\phi_0|\hat{P}_p|\phi_k}]\mathrm{Re}[\braket{\phi_0|\hat{P}_q|\phi_k}]\right].
\end{equation}
\end{widetext}
Comparing Eq.\eqref{eq:partial_variance_with_pauli_cov} and Eq.\eqref{eq:pauli_covariance_derive1} suggests defining the covariance as the expression within the bracket.
However, this definition violates the symmetric property of the covariance because $\hat{P}_p \hat{P}_q$ is not necessarily identical to $\hat{P}_q \hat{P}_p$.
Therefore, to ensure the symmetry, we superpose the product of operators in both orders, $(p, q)$ and $(q, p)$.
This results in the definition of the covariance as Eq.\eqref{eq:pauli_covariance_real}.

\section{Efficient Shifting Technique for Electronic Structure Problem}\label{sec:appendix_efficient_shifting}
Here, we describe an efficient procedure of optimizing Eq.\eqref{eq:min_norm} for an electronic structure Hamiltonian.
If $\hat{T}$ is given as Eq.\eqref{eq:shifting_operator_electronic_structure}, the number of real parameters determining $\hat{T}$ is: 
\begin{equation}
    N_{\mathrm{param}}=\left(N_{\mathrm{occ}}+N_{\mathrm{virt}}\right)\left(1+\binom{N_{\mathrm{orb}}-1}{2}\right),
\end{equation}
where $N_{\mathrm{occ}}$ and $N_{\mathrm{virt}}$ denote the number of occupied and virtual orbitals, respectively.
This number scales as $O(N_{\mathrm{orb}}^3)$ and makes the optimization computationally expensive.
In the rest of the section, we show an alternative and efficient method for the minimization of the norm by the shift technique.

We consider that the Hamiltonian $\hat{H}$ is given as an electronic structure Hamiltonian in Eq.\eqref{eq:electronic_structure_hamiltonian}, whose the indices are determined for the unique terms:
\begin{equation}\label{eq:electronic_structure_hamiltonian_precise}
    \hat{H} = \sum_{r\le s}^{N_{\mathrm{orb}}}{h_{rs}} \hat{E}_{rs} + \sum_{\substack{(p,q,r,s) \in \mathcal{P}}}g_{pqrs}(\hat{E}_{pq}\hat{E}_{rs}+\hat{E}_{rs}\hat{E}_{pq}),
\end{equation}
where
\begin{equation*}
\begin{split}
    \mathcal{P}:=\{(p,q,r,s):&(1\le p\le q \le N_{\mathrm{orb}})\\
    \wedge& (1 \le r \le s \le N_{\mathrm{orb}})\\
    \wedge& ((p < r) \vee ((p=r) \wedge (q \le s)))\\
    \wedge& \neg (p=q=r=s)\}.
\end{split}
\end{equation*}
The set $\mathcal{P}$ with the size of $|\mathcal{P}|\approx N_{\mathrm{orb}}^4/8$ represents the indices of Hermitian two-body operators avoiding the duplication from the following eight-fold symmetries:
\begin{equation*}
\resizebox{.96\hsize}{!}{$g_{pqrs}=g_{qprs}=g_{pqsr}=g_{qpsr}=g_{rspq}=g_{srpq}=g_{rsqp}=g_{srqp}$},
\end{equation*}
and squared number operators ($p=q=r=s$), which are absorbed into one-body number operators due to the idempotent property ($\hat{n}_p^2=\hat{n}_p$).

\begin{figure*}[t]
    \centering
    \subfloat[Pictorial analysis for a fixed $\bm{\mu^{\star}}$]{
        \label{fig:lin_solve_1}
        \includegraphics[width=0.3\linewidth]{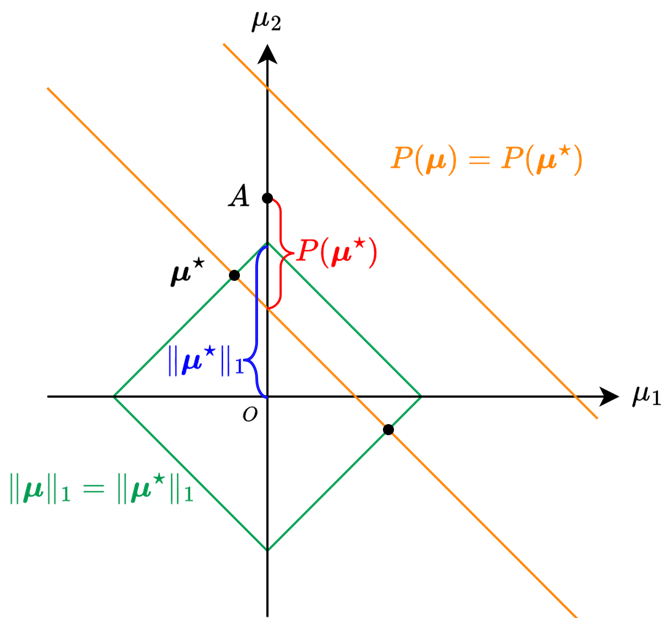}
    }
    \subfloat[Optima with a fixed 1-norm]{
        \label{fig:lin_solve_2}
        \includegraphics[width=0.3\linewidth]{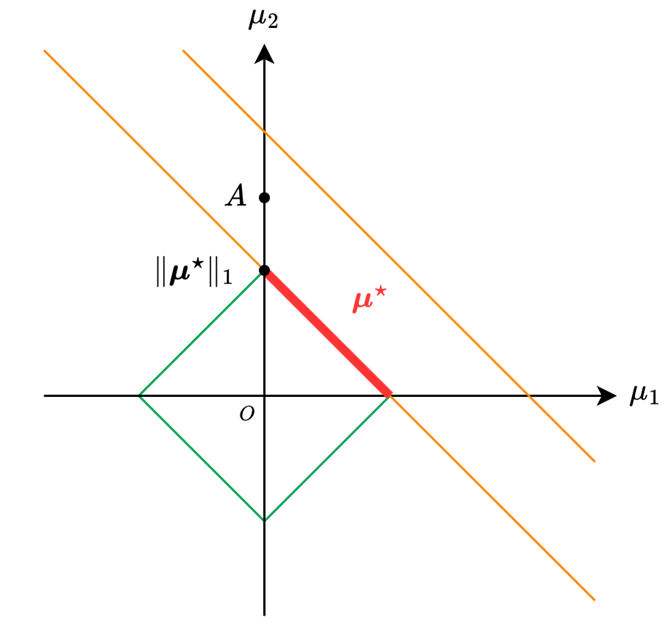}
    }
    \subfloat[Final optimal solutions]{
        \label{fig:lin_solve_3}
        \includegraphics[width=0.3\linewidth]{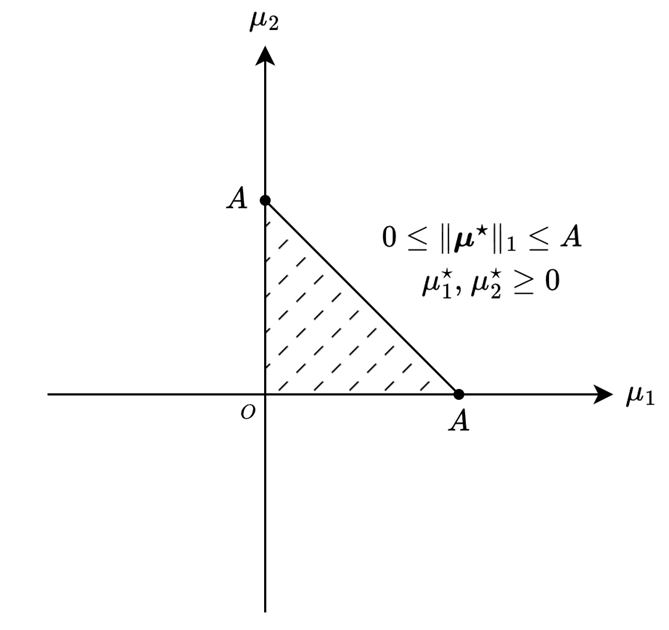}
    }
    \caption{A pictorial approach for the optimization of Eq.\eqref{eq:sep_opt_shift}.
    For the simplicity, the super- and subscripts $(rs)$ are omitted, and it is assumed that $\bm{\mu}$ is two-dimensional and $A>0$.
    In Figure (a), for a fixed $\bm{\mu^{\star}}$, the green diamond-shaped plot and the orange parallel lines denote all the points of $\bm{\mu}$ such that $\|\bm{\mu}\|_1=\|\bm{\mu^{\star}}\|_1$ and $P(\bm{\mu})=P(\bm{\mu^{\star}})$, respectively.
    It is shown that $P(\bm{\mu^{\star}})+\|\bm{\mu^{\star}}\|_1 \ge A$ must be satisfied for the intersections between two plots to be established, and thus $\bm{\mu}^{\star}$ exists.
    In Figure (b), for a fixed $\|\bm{\mu}\|_1$, the optimal points are found, as indicated by the red line, with the corresponding optimal value of $P(\bm{\mu^{\star}})+\|\bm{\mu^{\star}}\|_1 = A$.
    By varying $\|\bm{\mu}\|_1$, these points are expanded to the shaded area in Figure (c).
    For the case of $A\le0$, the plots are symmetrically transposed about the origin, resulting in the optimal points of $0\le \|\bm{\mu^{\star}}\|_1 \le -A$ and $\mu_1^\star, \mu_2^\star \le 0$.
    }
    \label{fig:lin_solve}
\end{figure*} 

In order to realize the efficient optimization, we find a part affected by the shift among the terms in Eq.\eqref{eq:electronic_structure_hamiltonian_precise}.
The terms relevant to Eq.\eqref{eq:shifting_operator_electronic_structure} are selected to construct the partial Hamiltonian, $\hat{H}^{(\mathrm{s})}$:
\begin{equation}\label{eq:effective_fragments}
\hat{H}^{(\mathrm{s})} = \sum_{r\in \mathcal{F}}2\tilde{h}_{rr}\hat{n}_r+\sum_{r<s}^{N_{\mathrm{orb}}}\tilde{h}_{rs}\hat{E}_{rs} + \sum_{q\in \mathcal{F}}\sum_{rs\in \mathcal{E}_{q}}\tilde{g}_{rsqq}\hat{E}_{rs}\hat{n}_q,
\end{equation}
where the modified coefficients are 
\begin{equation*}
\begin{split}
\tilde{h}_{rs}=&h_{rs}+\delta_{r\neq s}\sum_{q\in[N_{\mathrm{orb}}]\setminus\{r,s\}}g_{rqsq}\\        
\tilde{g}_{rsqq}=&4g_{rsqq}-2g_{rqsq}.
\end{split}
\end{equation*}
Such coefficients are determined by the following transformation to make the form of the Hamiltonian consistent to the shift operator:
\begin{equation}\label{eq:fact_fermion}
\begin{split}
    &g_{rqsq}\hat{E}_{rq}\hat{E}_{sq} + g_{rsqq}\hat{E}_{rs}\hat{E}_{qq} +\mathrm{h.c.} \\
    =& g_{rqsq}\hat{E}_{rs} +(4g_{rsqq}-2g_{rqsq})\hat{E}_{rs}\hat{n}_q,
\end{split}
\end{equation}
for all $r\neq q$ and $s \neq q$.
Because the rest of the Hamiltonian ($\hat{H}-\hat{H}^{(\mathrm{s})}$) is invariant to the shifting with respect to the term-wise grouping algorithm, denoted as $\mathtt{T}$, we focus on the following norm minimization:
\begin{equation}
    \min_{\bm{\tau}}\|\bm{\gamma}_{\mathtt{T}}(\hat{H}^{(\mathrm{s})}-\hat{T}(\bm{\tau}))\|.
\end{equation}
Here, $\hat{H}^{(\mathrm{s})} - \hat{T}(\bm{\tau})$ is given as
\begin{equation}\label{eq:shifted_effective_hamiltonian}
\begin{split}
    \hat{H}^{(\mathrm{s})}-\hat{T}(\bm{\tau}) =& \sum_{r\in\mathcal{F}}(2\tilde{h}_{rr} - \tau_r^{(1)})\hat{n}_r\\
    +&\sum_{r<s}(\tilde{h}_{rs}+\sum_{q\in{\mathrm{occ}}\setminus\{r,s\}}\tau_{qrs}^{(2)})\hat{E}_{rs}\\
    +&\sum_{q\in\mathcal{F}}\sum_{rs\in\mathcal{E}_q}(\tilde{g}_{rsqq}-\tau_{qrs}^{(2)})\hat{E}_{rs}\hat{n}_q.
\end{split}
\end{equation}
Note that the each operators in Eq.\eqref{eq:shifted_effective_hamiltonian} contributes to the norm as shown below:
\begin{equation}
\begin{gathered}
    \frac{1}{\sqrt{d}}\mathrm{Tr}\left[\hat{n}_r^2\right]^{1/2}=\frac{1}{\sqrt{2}},\\
    \frac{1}{\sqrt{d}}\mathrm{Tr}\left[\hat{E}_{rs}^2\right]^{1/2} = \sqrt{\frac{2}{d}}\mathrm{Tr}\left[(\hat{E}_{rs}\hat{n}_q)^2\right]^{1/2}=\sqrt{\frac{(3\delta_{rs}+1)}{2}},
\end{gathered}
\end{equation}
for $r\neq q$ and $s\neq q$.
Therefore, the norm with term-wise grouping is determined as
\begin{equation}\label{eq:shifted_norm_appdx}
\begin{split}
    \|\bm{\gamma}_{\mathtt{T}}(\hat{H}^{(\mathrm{s})} - \hat{T}(\bm{\tau}))\| =& \sum_{r\in \mathcal{F}}\frac{1}{\sqrt{2}}\left|2\tilde{h}_{rr}-\tau_{r}^{(1)}\right|\\
    +&\sum_{r<s}\frac{1}{\sqrt{2}}\left|\tilde{h}_{rs}+\sum_{q\in\mathrm{occ}\setminus\{r,s\}}\tau_{qrs}^{(2)}\right|\\
    +&\sum_{q\in \mathcal{F}}\sum_{rs\in\mathcal{E}_q}\sqrt{\frac{3\delta_{rs}+1}{2}}\left|\tilde{g}_{rsqq}-\tau_{qrs}^{(2)}\right|.
\end{split}
\end{equation}
The first summation becomes zero by assigning
\begin{equation}\label{eq:assign_tau_1}
    \tau_{r}^{(1)} = 2\tilde{h}_{rr} \quad\forall r\in \mathcal{F}.
\end{equation}
Furthermore, because the variables, $\tau_{qrs}^{(2)}$ for $(q,r,s)\in \mathcal{T}$ such that
\begin{equation}
\begin{split}
        \mathcal{T} &= \{(q,r,s) : 1\le r=s \le N_{\mathrm{orb}}, q\in\mathcal{F}\setminus\{r\}\}\\
    &\cup\{(q,r,s) : 1\le r<s \le N_{\mathrm{orb}}, q\in\mathrm{virt}\setminus\{r, s\}\},
\end{split}
\end{equation}
are only included in the third summation, they are also determined as 
\begin{equation}\label{eq:assign_tau_T}
    \tau_{qrs}^{(2)}=\tilde{g}_{rsqq} \quad \forall (q,r,s) \in \mathcal{T}.
\end{equation}
However, the rest of the variables, $\tau_{qrs}^{(2)}$ for $(q,r,s)\in \mathcal{T}^{\mathrm{c}}$ occur both in the second and the third summations, where
\begin{equation}
    \mathcal{T}^{\mathrm{c}} = \{(q,r,s): 1\le r < s \le N_{\mathrm{orb}}, q\in \mathrm{occ} \setminus \{r, s\}\}.
\end{equation}

After the assignment of Eqs.\eqref{eq:assign_tau_1} and \eqref{eq:assign_tau_T}, the minimization problem of Eq.\eqref{eq:shifted_norm_appdx} is then reduced to
\begin{widetext}
\begin{equation}\label{eq:reduced_shift}
\begin{split}
    \|\bm{\gamma}_{\mathtt{T}}(\hat{H}^{(\mathrm{s})} - \hat{T}'(\bm{\tau}^{(2)}_{\mathcal{T}^{\mathrm{c}}}))\| =& \frac{1}{\sqrt{2}}\sum_{r<s}\left[\left|\tilde{h}_{rs}+\sum_{q\in\mathrm{occ}\setminus\{r,s\}}\tau_{qrs}^{(2)}\right|+\sum_{q\in\mathrm{occ}\setminus\{r,s\}}\left|\tilde{g}_{rsqq}-\tau_{qrs}^{(2)}\right|\right]
\end{split}
\end{equation}    
\end{widetext}
where $\hat{T}'(\bm{\tau}^{(2)}_{\mathcal{T}^{\mathrm{c}}})$ denotes the shift operator $\hat{T}$ with the partial assignment and $\bm{\tau}^{(2)}_{\mathcal{T}^{\mathrm{c}}}$ denotes the set of variables, $\tau^{(2)}_{qrs}$ with $(q,r,s)\in \mathcal{T}^{\mathrm{c}}$.
Furthermore, minimizing Eq.\eqref{eq:reduced_shift} is identical to the separated optimizations:
\begin{equation}\label{eq:sep_opt_shift}
    \min_{\bm{\mu}^{(rs)}\in \mathbb{R}^{|\mathcal{U}_{rs}|}}(P_{rs}(\bm{\mu}^{(rs)}) + \|\bm{\mu}^{(rs)}\|_1)    
\end{equation}
for each $(r, s)$ with respect to the variables, $\bm{\mu}^{(rs)}=\{\mu_{q}^{(rs)}=\tilde{g}_{rsqq} - \tau_{qrs}^{(2)}: q\in \mathcal{U}_{rs}\}$, where $\mathcal{U}_{rs}=\mathrm{occ}\setminus\{r,s\}$ and
\begin{equation}
    P_{rs}(\bm{\mu}^{(rs)})=\left| A_{rs} - \sum_{q\in\mathcal{U}_{rs}}  \mu_{q}^{(rs)} \right|,
\end{equation}
and $A_{rs}=\tilde{h}_{rs}+\sum_{q\in \mathcal{U}_{rs}}\tilde{g}_{rsqq}$.
We provide a pictorial procedure to solve Eq.\eqref{eq:sep_opt_shift} in Fig.\ref{fig:lin_solve}.
Although the solution for 2-dimensional problem is shown, this can be extended to $|\mathcal{U}_{rs}|$-dimensional problem by considering hyperplane instead of lines in the figure.
Every $\bm{\tau}^{(2)\star}_{\mathcal{T}^{\mathrm{c}}}$ satisfying
\begin{equation}\label{eq:optimal_sep_opt}
\begin{split}
    \sum_{q\in\mathcal{U}_{rs}}|\tau_{qrs}^{(2)\star} - \tilde{g}_{rsqq}| &\le |A_{rs}|,\\
    \mathrm{Sign}(\tau_{qrs}^{(2)\star} - \tilde{g}_{rsqq})&=\mathrm{Sign}(A_{rs})
\end{split}
\end{equation}
for all $r$ and $s$, identically results in the optimal value of Eq.\eqref{eq:sep_opt_shift}, which is
\begin{equation}
\begin{split}
    \|\bm{\gamma}_{\mathtt{T}}(\hat{H}^{(\mathrm{s})} - \hat{T}'(\bm{\tau}_{\mathcal{T}^{\mathrm{c}}}^{(2)\star}))\|=
    \frac{1}{\sqrt{2}}\sum_{r<s}\left| \tilde{h}_{rs} + \sum_{q\in \mathrm{occ}\setminus\{r,s\}}\tilde{g}_{rsqq}. \right|
\end{split}
\end{equation}
Note that assigning $\tau_{qrs}^{(2)\star}=\tilde{g}_{rsqq}$ for all $r,s,q\in\mathcal{T}^{\mathrm{c}}$ satisfies Eq.\eqref{eq:optimal_sep_opt}, which is analogous to Eq.\eqref{eq:assign_tau_T}.

Therefore, we conclude that the parameters
\begin{align}
    \tau_{r}^{(1)} =& 2h_{rr}\\
    \tau_{qrs}^{(2)} =& 4g_{rsqq} - 2g_{rqsq}
\end{align}
lead to the optimal shift with respect to the term-wise grouping algorithm.
However, strictly to say, the above simplification only holds for the fragment Hamiltonian, as $\hat{n}_r$ and $\hat{E}_{rs}\hat{n}_q$ are Hermitian, not unitary.
Although $\hat{n}_q$ can be written as a unitary operator ($\hat{r}_{q}=2\hat{n}_q-1$), our current understanding of representing a unitary operator as a linear combination of one-body excitation operators remains insufficient to establish an analogous reduction for the LCU case.
Therefore, this work adopts the same shift operator for both LCU decomposition as an interim solution, leaving the parameter reduction for the LCU case as a future work.


\nocite{*}

\bibliography{revised}

\begin{thebibliography}{10}

\bibitem{Wintersperger2023-wc}
Wintersperger K, Dommert F, Ehmer T, Hoursanov A, Klepsch J, Mauerer W, et~al.
\newblock Neutral atom quantum computing hardware: performance and end-user
  perspective.
\newblock EPJ Quantum Technology. 2023 Aug;10(1):32.

\bibitem{Debnath2016}
Debnath S, Linke NM, Figgatt C, Landsman KA, Wright K, Monroe C.
\newblock Demonstration of a small programmable quantum computer with atomic
  qubits.
\newblock Nature. 2016 Aug;536(7614):63-6.

\bibitem{10.1063/1.4966970}
Brandl MF, van Mourik MW, Postler L, Nolf A, Lakhmanskiy K, Paiva RR, et~al.
\newblock {Cryogenic setup for trapped ion quantum computing}.
\newblock Review of Scientific Instruments. 2016 11;87(11):113103.

\bibitem{PhysRevLett.117.210502}
Wang XL, Chen LK, Li W, Huang HL, Liu C, Chen C, et~al.
\newblock Experimental Ten-Photon Entanglement.
\newblock Phys Rev Lett. 2016 Nov;117:210502.

\bibitem{Qiang2018}
Qiang X, Zhou X, Wang J, Wilkes CM, Loke T, O'Gara S, et~al.
\newblock Large-scale silicon quantum photonics implementing arbitrary
  two-qubit processing.
\newblock Nature Photonics. 2018 Sep;12(9):534-9.

\bibitem{PhysRevLett.109.060501}
Chow JM, Gambetta JM, C\'orcoles AD, Merkel ST, Smolin JA, Rigetti C, et~al.
\newblock Universal Quantum Gate Set Approaching Fault-Tolerant Thresholds with
  Superconducting Qubits.
\newblock Phys Rev Lett. 2012 Aug;109:060501.

\bibitem{Wendin_2017}
Wendin G.
\newblock Quantum information processing with superconducting circuits: a
  review.
\newblock Reports on Progress in Physics. 2017 sep;80(10):106001.

\bibitem{doi:10.1021/acs.chemrev.8b00803}
Cao Y, Romero J, Olson JP, Degroote M, Johnson PD, Kieferová M, et~al.
\newblock Quantum Chemistry in the Age of Quantum Computing.
\newblock Chemical Reviews. 2019;119(19):10856-915.
\newblock PMID: 31469277.

\bibitem{RevModPhys.92.015003}
McArdle S, Endo S, Aspuru-Guzik A, Benjamin SC, Yuan X.
\newblock Quantum computational chemistry.
\newblock Rev Mod Phys. 2020 Mar;92:015003.

\bibitem{doi:10.1021/acs.chemrev.9b00829}
Bauer B, Bravyi S, Motta M, Chan GKL.
\newblock Quantum Algorithms for Quantum Chemistry and Quantum Materials
  Science.
\newblock Chemical Reviews. 2020;120(22):12685-717.
\newblock PMID: 33090772.

\bibitem{Preskill2018quantumcomputingin}
Preskill J.
\newblock Quantum {C}omputing in the {NISQ} era and beyond.
\newblock {Quantum}. 2018 Aug;2:79.

\bibitem{RevModPhys.94.015004}
Bharti K, Cervera-Lierta A, Kyaw TH, Haug T, Alperin-Lea S, Anand A, et~al.
\newblock Noisy intermediate-scale quantum algorithms.
\newblock Rev Mod Phys. 2022 Feb;94:015004.

\bibitem{Cerezo2021}
Cerezo M, Arrasmith A, Babbush R, Benjamin SC, Endo S, Fujii K, et~al.
\newblock Variational quantum algorithms.
\newblock Nature Reviews Physics. 2021 Sep;3(9):625-44.

\bibitem{TILLY20221}
Tilly J, Chen H, Cao S, Picozzi D, Setia K, Li Y, et~al.
\newblock The Variational Quantum Eigensolver: A review of methods and best
  practices.
\newblock Physics Reports. 2022;986:1-128.
\newblock The Variational Quantum Eigensolver: a review of methods and best
  practices.

\bibitem{McClean2018}
McClean JR, Boixo S, Smelyanskiy VN, Babbush R, Neven H.
\newblock Barren plateaus in quantum neural network training landscapes.
\newblock Nature Communications. 2018 Nov;9(1):4812.

\bibitem{Cerezo2021_BP}
Cerezo M, Sone A, Volkoff T, Cincio L, Coles PJ.
\newblock Cost function dependent barren plateaus in shallow parametrized
  quantum circuits.
\newblock Nature Communications. 2021 Mar;12(1):1791.

\bibitem{Wang2021}
Wang S, Fontana E, Cerezo M, Sharma K, Sone A, Cincio L, et~al.
\newblock Noise-induced barren plateaus in variational quantum algorithms.
\newblock Nature Communications. 2021 Nov;12(1):6961.

\bibitem{Cerezo_2021_BP_derivative}
Cerezo M, Coles PJ.
\newblock Higher order derivatives of quantum neural networks with barren
  plateaus.
\newblock Quantum Science and Technology. 2021 jun;6(3):035006.

\bibitem{Arrasmith2021effectofbarren}
Arrasmith A, Cerezo M, Czarnik P, Cincio L, Coles PJ.
\newblock Effect of barren plateaus on gradient-free optimization.
\newblock {Quantum}. 2021 Oct;5:558.

\bibitem{Xue2022-kr}
Xue X, Russ M, Samkharadze N, Undseth B, Sammak A, Scappucci G, et~al.
\newblock Quantum logic with spin qubits crossing the surface code threshold.
\newblock Nature. 2022 Jan;601(7893):343-7.

\bibitem{Blume-Kohout2017-oj}
Blume-Kohout R, Gamble JK, Nielsen E, Rudinger K, Mizrahi J, Fortier K, et~al.
\newblock Demonstration of qubit operations below a rigorous fault tolerance
  threshold with gate set tomography.
\newblock Nature Communications. 2017 Feb;8(1):14485.

\bibitem{Postler2022-iv}
Postler L, Heu$\beta$en S, Pogorelov I, Rispler M, Feldker T, Meth M, et~al.
\newblock Demonstration of fault-tolerant universal quantum gate operations.
\newblock Nature. 2022 May;605(7911):675-80.

\bibitem{Bluvstein2024}
Bluvstein D, Evered SJ, Geim AA, Li SH, Zhou H, Manovitz T, et~al.
\newblock Logical quantum processor based on reconfigurable atom arrays.
\newblock Nature. 2024 Feb;626(7997):58-65.

\bibitem{katabarwa2023early}
Katabarwa A, Gratsea K, Caesura A, Johnson PD.
\newblock Early Fault-Tolerant Quantum Computing.
\newblock PRX Quantum. 2024 Jun;5:020101.

\bibitem{PhysRevLett.83.5162}
Abrams DS, Lloyd S.
\newblock Quantum Algorithm Providing Exponential Speed Increase for Finding
  Eigenvalues and Eigenvectors.
\newblock Phys Rev Lett. 1999 Dec;83:5162-5.

\bibitem{fractal_decomposition}
Suzuki M.
\newblock {Fractal decomposition of exponential operators with applications to
  many-body theories and Monte Carlo simulations}.
\newblock Physics Letters A. 1990;146(6):319-23.

\bibitem{Suzuki_book}
Hatano N, Suzuki M.
\newblock 2.
\newblock In: Das A, K~Chakrabarti B, editors. Finding Exponential Product
  Formulas of Higher Orders. Berlin, Heidelberg: Springer Berlin Heidelberg;
  2005. p. 37-68.

\bibitem{Somma_2019}
Somma RD.
\newblock Quantum eigenvalue estimation via time series analysis.
\newblock New Journal of Physics. 2019 dec;21(12):123025.

\bibitem{Zhang2022computingground}
Zhang R, Wang G, Johnson P.
\newblock Computing {G}round {S}tate {P}roperties with {E}arly
  {F}ault-{T}olerant {Q}uantum {C}omputers.
\newblock {Quantum}. 2022 Jul;6:761.

\bibitem{PRXQuantum.4.020331}
Ding Z, Lin L.
\newblock Even Shorter Quantum Circuit for Phase Estimation on Early
  Fault-Tolerant Quantum Computers with Applications to Ground-State Energy
  Estimation.
\newblock PRX Quantum. 2023 May;4:020331.

\bibitem{Ding2023simultaneous}
Ding Z, Lin L.
\newblock Simultaneous estimation of multiple eigenvalues with short-depth
  quantum circuit on early fault-tolerant quantum computers.
\newblock {Quantum}. 2023 Oct;7:1136.

\bibitem{parrish2019quantum}
Parrish RM, McMahon PL.
\newblock Quantum Filter Diagonalization: Quantum Eigendecomposition without
  Full Quantum Phase Estimation; 2019.
\newblock ArXiv:1909.08925 [quant-ph].

\bibitem{Stair2020-hq}
Stair NH, Huang R, Evangelista FA.
\newblock A Multireference Quantum Krylov Algorithm for Strongly Correlated
  Electrons.
\newblock J Chem Theory Comput. 2020 Apr;16(4):2236-45.

\bibitem{Kirby2023exactefficient}
Kirby W, Motta M, Mezzacapo A.
\newblock Exact and efficient {L}anczos method on a quantum computer.
\newblock {Quantum}. 2023 May;7:1018.

\bibitem{Theory_QKSD}
Epperly EN, Lin L, Nakatsukasa Y.
\newblock A Theory of Quantum Subspace Diagonalization.
\newblock SIAM Journal on Matrix Analysis and Applications. 2022;43(3):1263-90.

\bibitem{Lee2024samplingerror}
Lee G, Lee D, Huh J.
\newblock Sampling {E}rror {A}nalysis in {Q}uantum {K}rylov {S}ubspace
  {D}iagonalization.
\newblock {Quantum}. 2024 Sep;8:1477.

\bibitem{PRXQuantum.2.010333}
Seki K, Yunoki S.
\newblock Quantum Power Method by a Superposition of Time-Evolved States.
\newblock PRX Quantum. 2021 Feb;2:010333.

\bibitem{Motta2020-oi}
Motta M, Sun C, Tan ATK, O'Rourke MJ, Ye E, Minnich AJ, et~al.
\newblock Determining eigenstates and thermal states on a quantum computer
  using quantum imaginary time evolution.
\newblock Nature Physics. 2020 Feb;16(2):205-10.

\bibitem{Crawford2021efficientquantum}
Crawford O, Straaten Bv, Wang D, Parks T, Campbell E, Brierley S.
\newblock Efficient quantum measurement of {P}auli operators in the presence of
  finite sampling error.
\newblock {Quantum}. 2021 Jan;5:385.

\bibitem{Yen2023}
Yen TC, Ganeshram A, Izmaylov AF.
\newblock Deterministic improvements of quantum measurements with grouping of
  compatible operators, non-local transformations, and covariance estimates.
\newblock npj Quantum Information. 2023 Feb;9(1):14.

\bibitem{doi:10.1021/acs.jctc.3c00218}
Choi S, Izmaylov AF.
\newblock Measurement Optimization Techniques for Excited Electronic States in
  Near-Term Quantum Computing Algorithms.
\newblock Journal of Chemical Theory and Computation. 2023;19(11):3184-93.
\newblock PMID: 37224265.

\bibitem{doi:10.1021/acs.jctc.0c00008}
Yen TC, Verteletskyi V, Izmaylov AF.
\newblock Measuring All Compatible Operators in One Series of Single-Qubit
  Measurements Using Unitary Transformations.
\newblock Journal of Chemical Theory and Computation. 2020;16(4):2400-9.
\newblock PMID: 32150412.

\bibitem{doi:10.1021/acs.jctc.2c00837}
Choi S, Yen TC, Izmaylov AF.
\newblock Improving Quantum Measurements by Introducing “Ghost” Pauli
  Products.
\newblock Journal of Chemical Theory and Computation. 2022;18(12):7394-402.
\newblock PMID: 36332111.

\bibitem{Choi2023fluidfermionic}
Choi S, Loaiza I, Izmaylov AF.
\newblock Fluid fermionic fragments for optimizing quantum measurements of
  electronic {H}amiltonians in the variational quantum eigensolver.
\newblock {Quantum}. 2023 Jan;7:889.

\bibitem{PRXQuantum.2.040320}
Yen TC, Izmaylov AF.
\newblock Cartan Subalgebra Approach to Efficient Measurements of Quantum
  Observables.
\newblock PRX Quantum. 2021 Oct;2:040320.

\bibitem{10.1063/1.5141458}
Verteletskyi V, Yen TC, Izmaylov AF.
\newblock {Measurement optimization in the variational quantum eigensolver
  using a minimum clique cover}.
\newblock The Journal of Chemical Physics. 2020 03;152(12):124114.

\bibitem{Huang2020}
Huang HY, Kueng R, Preskill J.
\newblock Predicting many properties of a quantum system from very few
  measurements.
\newblock Nature Physics. 2020 Oct;16(10):1050-7.
\newblock Available from: \url{https://doi.org/10.1038/s41567-020-0932-7}.

\bibitem{hadfield2020lbcs}
Hadfield C, Bravyi S, Raymond R, Mezzacapo A.
\newblock Measurements of Quantum Hamiltonians with Locally-Biased Classical
  Shadows; 2020.
\newblock ArXiv:2006.15788 [quant-ph].

\bibitem{PhysRevLett.127.030503}
Huang HY, Kueng R, Preskill J.
\newblock Efficient Estimation of Pauli Observables by Derandomization.
\newblock Phys Rev Lett. 2021 Jul;127:030503.
\newblock Available from:
  \url{https://link.aps.org/doi/10.1103/PhysRevLett.127.030503}.

\bibitem{Nakamura_2024}
Nakamura Y, Yano Y, Yoshioka N.
\newblock Adaptive measurement strategy for quantum subspace methods.
\newblock New Journal of Physics. 2024 mar;26(3):033028.
\newblock Available from: \url{https://dx.doi.org/10.1088/1367-2630/ad2c3b}.

\bibitem{Loaiza_BLISS}
Loaiza I, Izmaylov AF.
\newblock Block-Invariant Symmetry Shift: Preprocessing Technique for
  Second-Quantized Hamiltonians to Improve Their Decompositions to Linear
  Combination of Unitaries.
\newblock Journal of Chemical Theory and Computation. 2023;19(22):8201-9.
\newblock PMID: 37939198.

\bibitem{Loaiza_2023}
Loaiza I, Khah AM, Wiebe N, Izmaylov AF.
\newblock Reducing molecular electronic Hamiltonian simulation cost for linear
  combination of unitaries approaches.
\newblock Quantum Science and Technology. 2023 may;8(3):035019.

\bibitem{LCU_childs}
Childs AM, Wiebe N.
\newblock Hamiltonian simulation using linear combinations of unitary
  operations.
\newblock Quantum Info Comput. 2012 nov;12(11–12):901–924.

\bibitem{LCU_Izmaylov}
Izmaylov AF, Yen TC, Lang RA, Verteletskyi V.
\newblock Unitary Partitioning Approach to the Measurement Problem in the
  Variational Quantum Eigensolver Method.
\newblock Journal of Chemical Theory and Computation. 2020;16(1):190-5.
\newblock PMID: 31747266.

\bibitem{LCU_Love}
Zhao A, Tranter A, Kirby WM, Ung SF, Miyake A, Love PJ.
\newblock Measurement reduction in variational quantum algorithms.
\newblock Phys Rev A. 2020 Jun;101:062322.

\bibitem{zheng2018depth}
Zheng YC, Lai CY, Brun TA, Kwek LC.
\newblock Depth reduction for quantum Clifford circuits through Pauli
  measurements; 2018.
\newblock ArXiv:1805.12082 [quant-ph].

\bibitem{gokhale2019minimizing}
Gokhale P, Angiuli O, Ding Y, Gui K, Tomesh T, Suchara M, et~al.
\newblock Minimizing State Preparations in Variational Quantum Eigensolver by
  Partitioning into Commuting Families; 2019.
\newblock ArXiv:1907.13623 [quant-ph].

\bibitem{bravyi_kitaev}
Seeley JT, Richard MJ, Love PJ.
\newblock {The Bravyi-Kitaev transformation for quantum computation of
  electronic structure}.
\newblock The Journal of Chemical Physics. 2012 12;137(22):224109.

\bibitem{qubit_tapering}
Bravyi S, Gambetta JM, Mezzacapo A, Temme K.
\newblock Tapering off qubits to simulate fermionic Hamiltonians; 2017.
\newblock ArXiv:1701.08213 [quant-ph].

\bibitem{yoshioka2024diagonalization}
Yoshioka N, Amico M, Kirby W, Jurcevic P, Dutt A, Fuller B, et~al.
\newblock Diagonalization of large many-body Hamiltonians on a quantum
  processor; 2024.
\newblock ArXiv:2407.14431 [quant-ph].

\bibitem{ofex}
Lee G.
\newblock Openfermion Expansion Toolkit; 2024.
\newblock GitHub repository, commit d735e3b,
  \url{https://github.com/snow0369/ofex}.

\bibitem{reproduce}
Lee G.
\newblock Measurement Strategies for QKSD; 2024.
\newblock GitHub repository, \url{https://github.com/snow0369/eff_qksd_script}.

\bibitem{https://doi.org/10.1002/jcd.21379}
Roy A, Suda S.
\newblock Complex Spherical Designs and Codes.
\newblock Journal of Combinatorial Designs. 2014;22(3):105-48.

\end{thebibliography}

\end{document}